\documentclass[english,final,table]{IEEEtran}
\usepackage[T1]{fontenc}
\usepackage[latin9]{inputenc}
\usepackage{xcolor}
\usepackage{pdfcolmk}
\usepackage{float}
\usepackage{mathrsfs}
\usepackage{bm}
\usepackage{amsmath}
\usepackage{amsthm}
\usepackage{amssymb}
\usepackage{graphicx}
\PassOptionsToPackage{normalem}{ulem}
\usepackage{ulem}

\makeatletter

\providecommand{\tabularnewline}{\\}
\floatstyle{ruled}
\newfloat{algorithm}{tbp}{loa}
\providecommand{\algorithmname}{Algorithm}
\floatname{algorithm}{\protect\algorithmname}
\providecolor{lyxadded}{rgb}{0,0,1}
\providecolor{lyxdeleted}{rgb}{1,0,0}
\DeclareRobustCommand{\lyxdeleted}[3]{{\color{lyxdeleted}\lyxsout{#3}}}
\DeclareRobustCommand{\lyxsout}[1]{\ifx\\#1\else\sout{#1}\fi}

\theoremstyle{plain}
\newtheorem{lem}{\protect\lemmaname}
\theoremstyle{plain}
\newtheorem{prop}{\protect\propositionname}

\usepackage{amsmath}
\usepackage{amssymb}
\usepackage[level=0]{wgroup_message}  % set level to 0, to hide all notes
\author{
\IEEEauthorblockN{Bowen~Li~and Junting~Chen}

\IEEEauthorblockA{School of Science and Engineering (SSE) and Future Network of Intelligence Institute (FNii) \\ The Chinese University of Hong Kong, Shenzhen, Guangdong 518172, China}
}


\usepackage[acronym]{glossaries}
\newcommand{\newac}{\newacronym}
\newcommand{\ac}{\gls}

\newcommand{\acpl}{\glspl}

\makeglossaries

\newac{speb}{SPEB}{square position error bound}
\newac[plural=EFIMs,firstplural=Fisher information matrices (EFIMs)]{efim}{EFIM}{Fisher information matrix}
\newac{ne}{NE}{Nash equilibrium}
\newac{mse}{MSE}{mean squared error}
\newac{toa}{TOA}{time-of-arrival}
\newac{snr}{SNR}{signal-to-noise ratio}
\newac{lan}{LAN}{local area network}
\newac{psd}{PSD}{positive semidefinite}
\newac{pd}{PD}{positive definite}
\newac{wrt}{w.r.t.}{with respect to}
\newac{lhs}{L.H.S.}{left hand side}
\newac{wp1}{w.p.1}{with probability 1}
\newac{kkt}{KKT}{Karush-Kuhn-Tucker}
\newac{wlog}{w.l.o.g.}{without loss of generality}
\newac{mle}{MLE}{maximum likelihood estimation}
\newac{gps}{GPS}{global positioning system}
\newac{rssi}{RSSI}{received signal strength indicator}
\newac{mimo}{MIMO}{multiple-input multiple-output}
\newac{csi}{CSI}{channel state information}
\newac{fdd}{FDD}{frequency division duplexing}
\newac{ms}{MS}{mobile station}
\newac{bs}{BS}{base station}
\newac{d2d}{D2D}{device-to-device}
\newac{slnr}{SLNR}{signal-to-interference-leakage-and-noise-ratio}
\newac{ula}{ULA}{uniform linear antenna array}
\newac{pas}{PAS}{power angular spectrum}
\newac{mmse}{MMSE}{minimum mean square error}
\newac{zf}{ZF}{zero-forcing}
\newac{rzf}{RZF}{regularized zero-forcing}
\newac{as}{AS}{angular spread}
\newac{aod}{AOD}{angle of departure}
\newac{iid}{i.i.d.}{independent and identically distributed} 
\newac{sinr}{SINR}{signal-to-interference-and-noise ratio}
\newac{tdd}{TDD}{time-division duplex}
\newac{rvq}{RVQ}{random vector quantization}
\newac{rhs}{R.H.S.}{right hand side}
\newac{mrc}{MRC}{maximum ratio combining}
\newac{cdf}{CDF}{cumulative distribution function}
\newac{a.s.}{a.s.}{almost surely}
\newac{los}{LOS}{line-of-sight}
\newac{jsdm}{JSDM}{joint spatial division and multiplexing}
\newac{map}{MAP}{maximum a posteriori}
\newac{klt}{KLT}{Karhunen-Lo\`eve Transform}
\newac{lbe}{LBE}{link bargaining equilibrium}
\newac{se}{SE}{Stackelberg equilibrium}
\newac{uav}{UAV}{unmanned aerial vehicle}
\newac{nlos}{NLOS}{non-line-of-sight}
\newac{pdf}{PDF}{probability density function}
\newac{em}{EM}{expectation-maximization}
\newac{knn}{KNN}{$k$-nearest neighbor}
\newac{svd}{SVD}{singular value decomposition}
\newac{nmf}{NMF}{non-negative matrix factorization}
\newac{umf}{UMF}{unimodality-constrained matrix factorization}
\newac{rmse}{RMSE}{rooted mean squared error}
\newac{olos}{OLOS}{obstructed line-of-sight}
\newac{mmw}{mmW}{millimeter wave}
\newac{ber}{BER}{bit error rate}
\newac{rss}{RSS}{received signal strength}
\newac{lp}{LP}{linear program}
\newac{ufw}{U-FW}{unimodal Frank-Wolfe}
\newac{utf}{UTF}{unimodality-constrained tensor factorization}
\newac{fw}{FW}{Frank-Wolfe}
\newac{iot}{IoT}{Internet-of-Things}
\newac{mae}{MAE}{mean absolute error}
\newac{crb}{CRB}{Cram\'er-Rao bound}
\newac{aoa}{AoA}{angle of arrival}
\newac{wcl}{WCL}{weighted centroid localization}

\setkeys{Gin}{width=1.0\columnwidth}

\renewcommand{\lyxdeleted}[3]{{\color{lyxdeleted}{}}}

\newac{mdp}{MDP}{Markov decision process}
\newac{wsn}{WSN}{wireless sensor network}
\newac{iff}{iff}{if and only if}
\newac{qsi}{QSI}{queue state information}
\newac{ue}{UE}{user equipment}
\newac{gbs}{GBS}{ground base station}
\newac{abs}{ABS}{aerial base station}
\newac{leo}{LEO}{low earth orbit}
\newac{gcs}{GCS}{ground control station}
\newac{tdma}{TDMA}{time division multiple access}
\newac{son}{SON}{self-organized network}
\newac{ofdm}{OFDM}{orthogonal frequency division multiplexing}
\newac{uma}{UMa}{urban macro}
\newac{isac}{ISAC}{integrated sensing and communication}
\newac{sca}{SCA}{successive convex approximation}
\newac{umi}{UMi}{Urban Micro}
\newac{lr}{LR}{Lagrangian relaxation}

\usepackage[noadjust]{cite}
\usepackage{color}  % make colors
\usepackage{graphicx,psfrag,cite,subfigure}

\makeatother

\usepackage{babel}
\providecommand{\lemmaname}{Lemma}
\providecommand{\propositionname}{Proposition}

\begin{document}
\title{Radio Map Assisted Approach for Interference-Aware Predictive UAV
Communications}
\maketitle
\begin{abstract}
Herein, an interference-aware predictive aerial-and-terrestrial communication
problem is studied, where an \ac{uav} delivers some data payload
to a few nodes within a communication deadline. The first challenge
is the possible interference to the ground \acpl{bs} and users possibly
at unknown locations. This paper develops a radio-map-based approach
to predict the channel to the receivers and the unintended nodes.
Therefore, a predictive communication strategy can be optimized ahead
of time to reduce the interference power and duration for the ground
nodes. Such predictive optimization raises the second challenge of
developing a low-complexity solution for a batch of transmission strategies
over $T$ time slots for $N$ receivers before the flight. Mathematically,
while the proposed interference-aware predictive communication problem
is non-convex, it is converted into a relaxed convex problem, and
solved by a novel dual-based algorithm, which is shown to achieve
global optimality at asymptotically small slot duration. The proposed
algorithm demonstrates orders of magnitude saving of the computational
time for moderate $T$ and $N$ compared to several existing solvers.
Simulations show that the radio-map-assisted scheme can prevent all
unintended receivers with known positions from experiencing interference
and significantly reduce the interference to the users at unknown
locations.
\end{abstract}


\section{Introduction\label{sec:intro}}

Low-altitude \ac{uav} activities have grown significantly over the
last decade \cite{WuXuZenNg:J21}. It is important to establish a
reliable communication network for real-time navigation, control,
and surveillance of the \ac{uav} network. For example, it may be
crucial to acquire the field of view of the \ac{uav} to monitor the
operation status during the mission of the \ac{uav}, and this requires
a reliable communication network.

One potential solution for \ac{uav} network communication is to construct
\acpl{son} that utilize the orthogonal spectrum of those used in
terrestrial cellular networks for interference avoidance \cite{YaoWanXuXu:J20,ShuSaaDaoNa:J20,MouGaoLiuWu:J22}.
However, \acpl{son} may not be reliable as the topology of the \ac{uav}
network can be significantly time-varying. Another viable solution
is to extend the coverage of terrestrial cellular networks to assist
with \ac{uav} communication. However, the transmission from the \ac{uav}
to ground \acpl{bs} may generate strong interference with other ground
terminals due to the high probability of \ac{los} conditions from
the sky \cite{ZhaZhaDiSon:J19,ChaSadBet:J19,ZhuGuoLiChe:J19,AzaGerGarPol:J20,HasKadAkh:J22,RahHosHeDai:J22,TanZhaHeZhu:J22,ZenLyuZha:J19}.

A majority of existing research that studied integrated \ac{uav}-and-terrestrial
communications ignored the interference \cite{ZhaZenZha:J18,JiYanSheXu:J20,LiChaCai:J20,MaZhoQiaChe:J21,LiChe:C22}
or assumed orthogonal transmissions between aerial nodes and ground
nodes \cite{HuCaiYuQin:J19,HuCaiLiuYu:J20,AlsYuk:J21}. Some recent
works \cite{ZhaZhaDiSon:J19,ChaSadBet:J19,ZhuGuoLiChe:J19,AzaGerGarPol:J20,HasKadAkh:J22,RahHosHeDai:J22,TanZhaHeZhu:J22,ZenLyuZha:J19}
attempted to mitigate the interference from the \acpl{uav} to the
ground nodes by optimizing the \ac{uav} trajectory, power control,
sub-channel allocation, and MIMO beamforming. However, these approaches
may not apply to some \ac{uav} networks where communication is not
the primary mission of the \acpl{uav}. First, the \ac{uav} trajectory
may not be altered for communication purposes during the mission,
and therefore, the approaches based on trajectory optimization \cite{ZhaZhaDiSon:J19,ChaSadBet:J19,ZhuGuoLiChe:J19}
are not suitable here. Secondly, as the topology of the \ac{uav}
network is time-varying, it becomes challenging to meet transmission
deadlines or age-of-information requirements for delivering large
and time-sensitive content, and these factors are not considered in
\cite{AzaGerGarPol:J20,HasKadAkh:J22,RahHosHeDai:J22,TanZhaHeZhu:J22}.

This paper studies an integrated aerial-and-terrestrial communication
scenario, where a \ac{uav} node uses the cellular spectrum to transmit
to $N$ nodes, which can be other \ac{uav} nodes or ground \acpl{bs}
with known locations. The objective is to deliver a given amount of
data within a communication deadline, resulting in a planning problem
for the communication timing and resource allocation. The main challenge
is the air-to-ground interference to the unintended \acpl{bs} and
users on the ground due to the transmission of the \ac{uav}. Specifically,
a general interference mitigation may require the \ac{csi} of the
unintended nodes for the entire transmission, which can be challenging
for this resource planning problem where the future \ac{csi} cannot
be obtained via online measurements. In addition, while the locations
of the \acpl{bs} can be known by the network for a rough \ac{csi}
prediction, the locations of the ground users are usually unknown.
As for a universal interference management strategy for such unknown
receivers, it is desired to transmit with minimum energy and duration,
which makes the resource planning and optimization problem more difficult
to solve.

To tackle the interference to the nodes with known locations, we employ
a {\em radio map} approach, where the radio map captures the large-scale
channel information between any two locations, including path loss,
shadowing, and statistics of the small-scale fading; but the actual
channel is not available on the radio map. The opportunity that drives
the radio map approach is the fact that aerial nodes mostly have predetermined
trajectories. For example, the trajectories of cargo delivery \acpl{uav}
and many patrol surveillance \acpl{uav} are determined by the operators
prior to the flight. As a consequence, the future \ac{csi} of the
nodes with known locations can be predicted by the \ac{uav} trajectory
using the radio map, and therefore, one can optimize for the {\em predictive UAV communications}.
To tackle the interference to the unknown ground nodes, we introduce
a sleep strategy aiming at reducing the transmission duration and
lowering the chance of interfering with the ground nodes with unknown
channel status. To summarize, this paper attempts to address the following
challenges:
\begin{itemize}
\item {\em How to exploit radio maps to optimize for interference-aware communication.}
First, as radio maps may only capture large-scale \ac{csi}, the uncertainty
of the future channel quality in predictive communication needs to
be addressed. In addition, the impact of the air-to-ground interference
for ground nodes at unknown locations needs to be considered.
\item {\em How to develop low-complexity solutions over a large horizon for the predictive communication.}
The predictive \ac{uav} communication problem involves determining
a sequence of transmission strategies over $T$ time slots {\em before}
the flight. Therefore, optimizing the solution based on standard off-the-shelf
solvers can be computationally prohibitive for a large $T$ even for
convex problems.
\end{itemize}
\ \ \ Mathematically, we formulate a radio-map-assisted predictive
communication problem for $N$-receiver interference-aware \ac{uav}
communications with a delay constraint. While previous research has
considered the impact of small-scale fading on transmission strategy
design, for example, prediction of the future small-scale fading \cite{YouZha:J19},
online adjustment based on real-time small-scale fading \cite{SamShaAssNgu:J20},
and constant attenuation on the channel \cite{LiZhaZhaYan:J21}, these
approaches cannot be applied to the transmission planning in predictive
communications. To address the challenge of the discontinuity due
to the on-off control for the interference-aware sleep mode optimization,
we propose a relax-then-round algorithm for an efficient solution
with asymptotic optimality guarantee. Moreover, we develop a dual-based
algorithm, substantially reducing the computational complexity without
sacrificing the optimality. In summary, we make the following contributions:

\begin{itemize}
\item We formulate an interference-aware predictive communication problem
exploiting radio maps. In this formulation, we address the uncertainty
of the future \ac{csi} by developing a deterministic expected capacity
lower bound, and mitigate the air-to-ground interference for ground
nodes at unknown locations by penalizing the transmission duration.
\item While the problem is non-convex, we develop a relax-then-round optimization
strategy, which is proven to have the asymptotic optimality guarantee.
Based on the Lagrangian dual technique, we develop a low-complexity
algorithm, which achieves orders of magnitude of complexity savings,
compared to state-of-the-art \ac{lr}-based and \ac{sca}-based algorithms
for a large horizon $T$.
\item Our simulations show that the proposed radio-map-assisted scheme can
prevent all unintended receivers with known positions from experiencing
interference and also significantly reduce the interference to the
users at unknown locations. Moreover, the proposed relax-then-round
can achieve the global optimality at asymptotically small slot duration,
and the running time is 1000 times less than the \ac{lr} \& \ac{sca}
scheme.
\end{itemize}
\ \ \ The rest of the paper is organized as follows. Section \ref{sec:System-Model}
presents the communication system model and the problem formulation.
Section \ref{sec:Problem-Reformulation} reformulates the problem
to handle the uncertainty from the small-scale fading then relaxes
the integer programming problem to a convex problem and propose a
dual-based algorithm. We develop a rounding strategy to modify the
relaxed results for suitable the original problem in Section \ref{sec:Rounding-Policy}.
Numerical results are demonstrated in Section \ref{sec:Simulation},
and conclusions are given in Section \ref{sec:Conclusion}.

\begin{figure}
\begin{centering}
\includegraphics[width=0.4\textwidth]{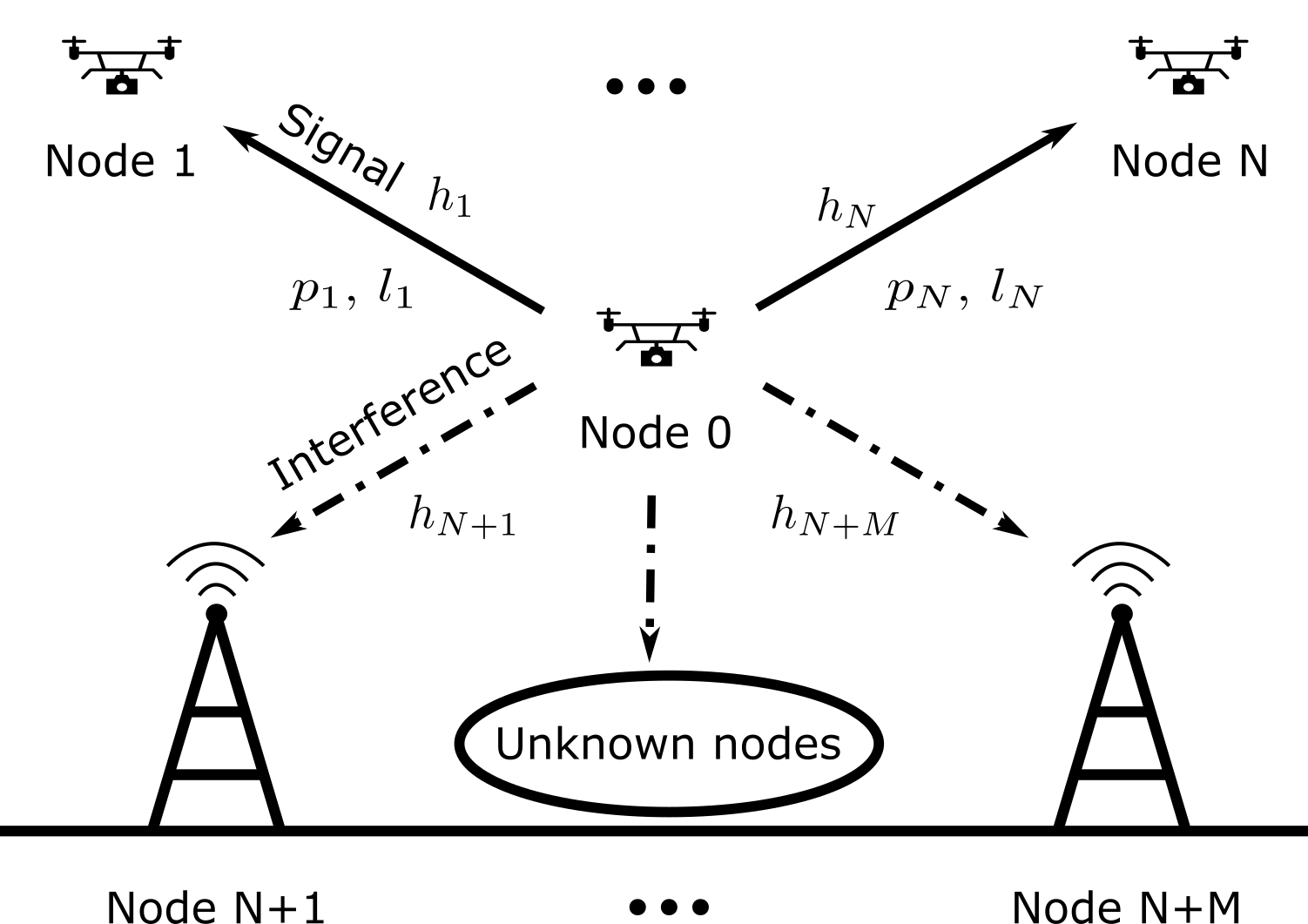}
\par\end{centering}
\caption{Slotted communication system illustration. \label{fig:system_model}}
\end{figure}

\section{System Model\label{sec:System-Model}}

Consider a multiuser transmission system with a single transmitter,
$N$ receiver nodes, and $M$ neighboring nodes that may be interfered
by the transmitter, as shown in Fig. \ref{fig:system_model}. The
positions of the transmitter, receivers, and the neighboring nodes
are known and denoted as $\bm{q}_{j}\in\mathbb{R}^{3}$, where $j=0,1,\dots,N+M$.
Here, $j=0$ refers to the transmitter, while $1\leq j\leq N$ refers
to the receiver nodes, and $N+1\leq j\leq N+M$ refers to the neighboring
nodes.

\subsection{Channel and radio map models}

We consider a flat fading channel model, in which the channel power
gain from the transmitter to node $1\le j\le N+M$ is given by 
\begin{equation}
h_{j}=g_{j}\xi_{j}\label{eq:channel_model}
\end{equation}
where $g_{j}$ and $\xi_{j}$ represent the channel gain due to the
large- and small-scale fading, respectively. The small-scale fading
$\xi_{j}$ is random and follows a Gamma distribution $\text{G}\left(\kappa_{j},1/\kappa_{j}\right)$
with shape parameter $\kappa_{j}$, and the scale parameter is chosen
as $1/\kappa_{j}$ to normalize the mean of $\xi_{j}$ to 1. Note
that Gamma distribution can be used to model various fading models,
including Rayleigh and Nakagami. For example, the Gamma distribution
with $\kappa_{j}=1$ degenerates to an exponential distribution, which
corresponds to the power gain of the Rayleigh fading channel.

Assume that the system has access to the large-scale channel parameters,
including the large-scale channel gain $g_{j}$ and the Gamma distribution
parameter $\kappa_{j}$, via a radio map $\bm{\Theta}(\bm{q}_{0},\bm{q}_{j})=(g_{j},\kappa_{j})$
that records the large-scale channel parameters $g_{j}$ and $\kappa_{j}$
as a function of the transmitter and receiver locations $\bm{q}_{0}$
and $\bm{q}_{j}$ \cite{LevYapKutCai:J21,liuche:J23}. As a consequence,
the statistics of the channel $h_{j}$ is available ahead of time
when the trajectories of the nodes $\bm{q}_{j}$ are available.

\subsection{Communication model}

Consider slotted transmission, where in each time slot $t\in\mathcal{T}\triangleq\{1,2,\cdots,T\}$,
a fraction of non-overlapping frequency resource $l_{n}(t)$ is allocated
to the $n$th receiver, $n\in\mathcal{N}\triangleq\{1,2,\dots,N\}$.
Denote the frequency allocation on time $t$ as $\boldsymbol{l}(t)\in\mathcal{L}$,
where
\[
\mathcal{L}\triangleq\left\{ \left[l_{n}\left(t\right)\right]_{n\in\mathcal{N}}:l_{n}\left(t\right)\in\left[0,1\right],\sum_{n\in\mathcal{N}}l_{n}\left(t\right)\in\left[0,1\right]\right\} 
\]
denotes the set of feasible frequency resources, in which the total
resources are normalized to 1. As a result, the total transmission
duration among $T$ time slots can be computed as $\sum_{t\in\mathcal{T}}\mathbb{I}\{\sum_{n\in\mathcal{N}}l_{n}(t)>0\}$.

Let $p_{n}(t)$ be the power allocated to receiver $n$ at time slot
$t$, where $p_{n}(t)\in\mathcal{P}\triangleq\{p:p\geq0\}$. As orthogonal
frequency resources are used for the $N$ receivers, the total normalized
throughput for the $n$th receiver over $T$ time slots can be computed
as
\begin{equation}
\Upsilon_{n}\left(\bm{p}_{n},\bm{l}_{n}\right)=\sum_{t\in\mathcal{T}}\log\left(1+p_{n}\left(t\right)h_{n}\left(t\right)\right)l_{n}\left(t\right)\label{eq:thp_n_def}
\end{equation}
where $\bm{p}_{n}=[p_{n}(t)]_{t\in\mathcal{T}}\in\mathbb{R}^{T\times1}$
and $\bm{l}_{n}=[l_{n}(t)]_{t\in\mathcal{T}}\in\mathbb{R}^{T\times1}$.

Meanwhile, the maximal interference over all frequency resources for
the $m$th neighboring node at time slot $t$ is modeled as
\[
I_{m}\left(t\right)=h_{m}\left(t\right)\cdot\max\left\{ p_{n}\left(t\right):n\in\mathcal{N}\right\} .
\]

Note that there could be other ground nodes at unknown locations and
their channels are not available at the transmitter node. To reduce
the impact on those unknown ground nodes, we propose to minimize both
the transmission power and the total transmission time, resulting
in the following cost function
\begin{equation}
F\left(\bm{P},\bm{L}\right)=\sum_{t\in\mathcal{T}}\sum_{n\in\mathcal{N}}p_{n}\left(t\right)l_{n}\left(t\right)+\lambda\sum_{t\in\mathcal{T}}\mathbb{I}\left\{ \sum_{n\in\mathcal{N}}l_{n}\left(t\right)>0\right\} \label{eq:erg_def}
\end{equation}
where $\boldsymbol{P}=[\bm{p}_{n}]_{n\in\mathcal{N}}\in\mathbb{R}^{T\times N}$,
$\boldsymbol{L}=[\bm{l}_{n}]_{n\in\mathcal{N}}\in\mathbb{R}^{T\times N}$,
and $\mathbb{I}\{\mathcal{A}\}$ is an indicator function defined
as $\mathbb{I}\{\mathcal{A}\}=1$ if condition $\mathcal{A}$ is satisfied,
and $\mathbb{I}\{\mathcal{A}\}=0$ otherwise. The first term in (\ref{eq:erg_def})
captures the transmission power, the second term captures the total
transmission time, and $\lambda$ is the weighting factor for the
two terms.

This paper aims to plan the transmission strategy $\bm{P}$ and $\boldsymbol{L}$
for $T$ time slots ahead for the transmission that minimizes the
overall interference and controls interference to the neighboring
nodes less than $I_{\text{bs}}$ while delivering $S_{n}$ bits of
data to $N$ receivers
\begin{align}
\underset{\bm{P},\bm{L}}{\text{minimize}} & \quad F\left(\bm{P},\bm{L}\right)\nonumber \\
\text{subject to} & \quad\mathbb{E}\left\{ \Upsilon_{n}\left(\bm{p}_{n},\bm{l}_{n}\right)\right\} \ge S_{n},\forall n\label{eq:throughput_constraint_v0}\\
 & \quad\mathbb{E}\left\{ I_{m}\left(t\right)\right\} \le I_{\text{bs}},\forall m,t\label{eq:ubiquitous_interference_constraint_v0}\\
 & \quad p_{n}\left(t\right)\in\mathcal{P},\forall n,t,\,\boldsymbol{l}\left(t\right)\in\mathcal{L},\forall t\label{eq:resources_constraint_v0}
\end{align}
where (\ref{eq:throughput_constraint_v0}) and (\ref{eq:ubiquitous_interference_constraint_v0})
are the expected throughput constraints and the expected interference
constraints, respectively.

\section{Radio-Map-Assisted Optimization\label{sec:Problem-Reformulation}}

The radio-map-assisted optimization faces a significant challenge
in dealing with channel uncertainty, as only large-scale information
is available. This section addresses this challenge by analyzing the
expected throughput and interference using large-scale information
from the radio map. A robust formulation is thus constructed, then
transformed and relaxed into a convex problem that enables efficient
and effective solution techniques to be applied. Furthermore, we propose
an efficient dual-based algorithm to accelerate the search for optimal
solutions.

\subsection{Radio-map-assisted reformulation}

With the aid of the radio map $\bm{\Theta}$ and the known locations
of the transmitter, receivers, and neighboring nodes at time $t$,
$\bm{q}_{j}(t)$, the large-scale channel parameters at time $t$
are $(g_{j}(t),\kappa_{j}(t))=\bm{\Theta}(\bm{q}_{0}(t),\bm{q}_{j}(t))$.
Thus, the power gain at time $t$, $h_{j}(t)$ defined in (\ref{eq:channel_model}),
follows Gamma distribution $\text{G}\left(\kappa_{j}(t),g_{j}(t)/\kappa_{j}(t)\right)$.
Based on the statistical information of $h_{n}(t)$, $n\in\mathcal{N}$,
the expected channel capacity can be explicitly lower bounded as follows.
\begin{lem}
\label{lem:Lower_bound_of_thp_exp}(A deterministic capacity lower
bound) The expected capacity $\mathbb{E}\{\log(1+p_{n}h_{n})\}$ is
lower bounded by $\text{\ensuremath{\log\left(1+p_{n}g_{n}\right)}}-\epsilon_{n}$
where $\epsilon_{n}=\log(e)/\kappa_{n}-\log(1+(2\kappa_{n})^{-1}).$
\end{lem}
\begin{proof}
See Appendix \ref{sec:Proof_lem_lower_bound_of_thp_exp}.
\end{proof}
The gap $\epsilon_{n}$ is always positive, and tends to 0 when $\kappa_{n}$
goes to infinity where the channel is asymptotically deterministic.

As a result of Lemma \ref{lem:Lower_bound_of_thp_exp}, the expected
throughput $\mathbb{E}\left\{ \Upsilon_{n}\left(\bm{p}_{n},\bm{l}_{n}\right)\right\} $
is lower bounded by
\[
\mathbb{E}\left\{ \Upsilon_{n}\left(\bm{p}_{n},\bm{l}_{n}\right)\right\} \ge\sum_{t\in\mathcal{T}}\left(\text{\ensuremath{\log\left(1+p_{n}\left(t\right)g_{n}\left(t\right)\right)}}-\epsilon_{n}\right)l_{n}\left(t\right).
\]
In other words, the expected throughput constraint (\ref{eq:throughput_constraint_v0})
can be relaxed to 
\begin{equation}
\sum_{t\in\mathcal{T}}c_{n}\left(t\right)l_{n}\left(t\right)\ge S_{n},\forall n\label{eq:rlx_exp_thp_c}
\end{equation}
where $c_{n}(t)=\text{\ensuremath{\log\left(1+p_{n}\left(t\right)g_{n}\left(t\right)\right)}}-\epsilon_{n}(t)$
is the approximated channel capacity.

In addition, using the statistical information of $h_{m}(t)$, $m\in\mathcal{M}\triangleq\{N+1,\cdots,N+M\}$,
the expected interference $\mathbb{E}\{I_{m}\}$ at node $m$ can
be computed as $\mathbb{E}\left\{ I_{m}\right\} =\mathbb{E}\left\{ h_{m}\left(t\right)\cdot\max\left\{ p_{n}\left(t\right):n\in\mathcal{N}\right\} \right\} =g_{m}\cdot\max\left\{ p_{n}\left(t\right):n\in\mathcal{N}\right\} .$
As a result, the interference constraint (\ref{eq:ubiquitous_interference_constraint_v0})
can be equivalently transformed into a power constraint as
\begin{equation}
p_{n}\left(t\right)\le\bar{p}\left(t\right)\triangleq I_{\text{bs}}/\max_{m\in\mathcal{M}}g_{m}\left(t\right),\,\forall t.\label{eq:power_upper_constraint}
\end{equation}

Using the approximated throughput constraint (\ref{eq:rlx_exp_thp_c})
and interference-equivalent power constraint (\ref{eq:power_upper_constraint}),
the original problem can be relaxed into the following problem
\begin{align}
\mathscr{P}1:\quad\underset{\bm{P},\bm{L}}{\text{minimize}} & \quad F\left(\bm{P},\bm{L}\right)\nonumber \\
\text{subject to} & \quad\sum_{t\in\mathcal{T}}c_{n}\left(t\right)l_{n}\left(t\right)\ge S_{n},\forall n\label{eq:thp_c_v2}\\
 & \quad p_{n}\left(t\right)\in\mathcal{P}^{\prime},\forall n,t,\,\boldsymbol{l}\left(t\right)\in\mathcal{L},\forall t\label{eq:resources_constraint_v2}
\end{align}
where $\mathcal{P}^{\prime}\triangleq\{p(t):0\le p(t)\le\bar{p}(t)\}$.

However, this problem is non-convex due to the indicator function
in the objective function in (\ref{eq:erg_def}). The state-of-the-art
techniques to handle the indicator function include \ac{lr} \cite{Fis:J81}
and \ac{sca} \cite{Raz:T14}, but, the optimality is unclear, and
the computational complexity can be high.

In the rest of the paper, we develop a cost-aware relax-then-round
scheme to solve problem $\mathscr{P}1$. First, the indicator functions
$\mathbb{I}\{\sum_{n\in\mathcal{N}}l_{n}(t)>0\}$ in the objective
(\ref{eq:erg_def}) are relaxed to continuous ones, then we transform
the non-convex problem to a convex equivalent problem and propose
an efficient algorithm. Subsequently in Section \ref{sec:Rounding-Policy},
a cost-aware rounding strategy is developed to ensure asymptotic optimality.

\subsection{Convex relaxation}

\subsubsection{Relaxation from indicator function\label{subsec:rlx_from_IP}}

Recall that the presence of the indicator function in the objective
(\ref{eq:erg_def}) implies the need for integer programming to optimize
the on-off mode. We replace the indicator $\mathbb{I}\{\sum_{n\in\mathcal{N}}l_{n}(t)>0\}$
by a continuous variable $\sum_{n\in\mathcal{N}}l_{n}(t)$ which can
take continuous value from 0 to 1. This relaxation results in the
following objective
\begin{equation}
\tilde{F}\left(\bm{P},\bm{L}\right)=\sum_{t\in\mathcal{T}}\sum_{n\in\mathcal{N}}\left(p_{n}\left(t\right)+\lambda\right)l_{n}\left(t\right).\label{eq:erg_relax}
\end{equation}

It is clear that the relaxed objective (\ref{eq:erg_relax}) gives
a lower bound of the cost function (\ref{eq:erg_def}), because $\sum_{n\in\mathcal{N}}l_{n}(t)\le\mathbb{I}\{\sum_{n\in\mathcal{N}}l_{n}(t)>0\}$
with equality for $\sum_{n\in\mathcal{N}}l_{n}(t)=0$ or $1$. Such
a property motivates our rounding strategy which will be discussed
in Section \ref{sec:Rounding-Policy}. As a result, problem $\mathscr{P}1$
becomes 
\[
\mathscr{P}2:\quad\underset{\bm{P},\bm{L}}{\text{minimize}}\quad\tilde{F}\left(\bm{P},\bm{L}\right),\quad\text{subject to}\quad\text{(\ref{eq:thp_c_v2}), (\ref{eq:resources_constraint_v2})}.
\]

Note that the optimal value of problem $\mathscr{P}2$ is a lower
bound on the optimal value of $\mathscr{P}1$, because $\tilde{F}(\bm{P},\bm{L})\le F(\bm{P},\bm{L})$
for all $(\bm{P},\bm{L})$ and the feasible sets of the two problems
are the same.

\subsubsection{Convex transformation\label{subsec:Convex-transformation}}

To handle the non-convex constraint (\ref{eq:thp_c_v2}), we introduce
a new variable
\begin{equation}
\phi_{n}\left(t\right)\triangleq\left(\text{\ensuremath{\log\left(1+p_{n}\left(t\right)g_{n}\left(t\right)\right)}}-\epsilon_{n}\left(t\right)\right)l_{n}\left(t\right)\label{eq:local_thp_def}
\end{equation}
to replace the variable $p_{n}(t)$. As a result, the throughput constraint
(\ref{eq:thp_c_v2}) becomes linear
\begin{equation}
\sum_{t\in\mathcal{T}}\phi_{n}\left(t\right)\ge S_{n},\,\forall n\label{eq:thp_trans}
\end{equation}
and the power constant in (\ref{eq:resources_constraint_v2}) also
becomes linear
\begin{equation}
\phi_{n}\left(t\right)\in\left[-\epsilon_{n}\left(t\right)l_{n}\left(t\right),\bar{c}_{n}\left(t\right)l_{n}\left(t\right)\right]\label{eq:lc_trans}
\end{equation}
where $\bar{c}_{n}(t)\triangleq\log(1+\bar{p}_{n}(t)g_{n}(t))-\epsilon_{n}(t)$,
due to the monotonic increasing property of the function $\phi_{n}(t)$
over $p_{n}(t)$. In addition, the variable $p_{n}(t)$ can be derived
to a function of $\phi_{n}(t)$ and $l_{n}(t)$, as 
\begin{equation}
p_{n}\left(t\right)=\begin{cases}
\left(2^{\phi_{n}\left(t\right)/l_{n}\left(t\right)+\epsilon_{n}\left(t\right)}-1\right)/g_{n}\left(t\right) & l_{n}\left(t\right)>0\\
0 & l_{n}\left(t\right)=0.
\end{cases}\label{eq:p_over_phi_l}
\end{equation}
As a result, the cost function (\ref{eq:erg_relax}) can be expressed
as 
\begin{equation}
\sum_{t\in\mathcal{T}}\sum_{n\in\mathcal{N}}\left(\frac{2^{\frac{\phi_{n}\left(t\right)}{l_{n}\left(t\right)}+\epsilon_{n}\left(t\right)}-1}{g_{n}\left(t\right)}+\lambda\right)l_{n}\left(t\right)\label{eq:erg_trans}
\end{equation}
which is convex, because for each $t$ and $n$, the summand $((2^{\phi_{n}(t)/l_{n}(t)+\epsilon_{n}(t)}-1)/g_{n}(t)+\lambda)l_{n}(t)$
is a perspective function of a convex function $\varphi(x)\triangleq(2^{x+\epsilon_{n}(t)}-1)/g_{n}(t)+\lambda$,
and thus is convex \cite{boyd2004convex}.

Using the transformed objective function (\ref{eq:erg_trans}), and
the transformed constraints (\ref{eq:thp_trans})\textendash (\ref{eq:lc_trans}),
the problem $\mathscr{P}2$ is transformed into the following problem
\begin{align*}
\mathscr{P}3:\quad\underset{\bm{\Phi},\bm{L}}{\text{minimize}} & \quad\sum_{t\in\mathcal{T}}\sum_{n\in\mathcal{N}}\left(\frac{2^{\frac{\phi_{n}\left(t\right)}{l_{n}\left(t\right)}+\epsilon_{n}\left(t\right)}-1}{g_{n}\left(t\right)}+\lambda\right)l_{n}\left(t\right)\\
\text{subject to} & \quad\sum_{t\in\mathcal{T}}\phi_{n}\left(t\right)\ge S_{n},\forall n\\
 & \quad\phi_{n}\left(t\right)\in\left[-\epsilon_{n}\left(t\right)l_{n}\left(t\right),\bar{c}_{n}\left(t\right)l_{n}\left(t\right)\right],\forall n,t\\
 & \quad\boldsymbol{l}\left(t\right)\in\mathcal{L},\forall t
\end{align*}
where $\bm{\Phi}\triangleq[\phi_{n}\left(t\right)]_{t\in\mathcal{T},n\in\mathcal{N}}\in\mathbb{R}^{T\times N}$.

The transformed problem $\mathscr{P}3$ is convex because the objective
function is convex and all the constraints are linear. Therefore,
it can be solved efficiently.

\subsection{Optimal solution for $N=1$\label{sec:Dual_algorithm}}

While problem $\mathscr{P}3$ is convex, solving problem $\mathscr{P}3$
by general solvers, such as the (iterative) subgradient or interior-point
methods, has complexity higher than $\mathcal{O}(N^{3}T^{3})$ per
iteration \cite{Ber:b15}, which will be overwhelming for a large
number of time slots $T$ when we want to optimize over a large time
horizon or for smaller time slots. As a result, the key challenge
in radio-map-assisted predictive communication is to develop efficient
algorithms to solve the optimization problem for a large $T$. To
tackle this challenge, we first develop a semi-closed form solution
to $\mathscr{P}3$ under $N=1$.

For simplicity, we omit the subscript $n$ in the variables such as
$\phi_{n}$ and $l_{n}$ since $N=1$. Then in problem $\mathscr{P}3$,
the resource allocation variables become $\bm{\phi}=[\phi(t)]_{t\in\mathcal{T}}$
and $\bm{l}=[l(t)]_{t\in\mathcal{T}}$, and the cost is written as
$\sum_{t\in\mathcal{T}}((2^{\phi(t)/l(t)+\epsilon(t)}-1)/g(t)+\lambda)l(t)$.
Since $\mathscr{P}3$ is convex, the solution to $\mathscr{P}3$ satisfies
the \ac{kkt} conditions which can be derived as follows.
\begin{lem}
\label{lem:unique_root}let $x_{0}$ be the solution to $\vartheta\left(x\right)=0$,
where 
\begin{equation}
\vartheta\left(x\right)\triangleq\big(2^{x+\epsilon\left(t\right)}-1\big)/g\left(t\right)+\lambda-x\big(\ln2\cdot2^{x+\epsilon\left(t\right)}/g\left(t\right)\big).\label{eq:def_theta_c}
\end{equation}
The root $x_{0}$ is unique in the region $x\geq0$, Moreover, $0<x_{0}\le\max\{2/\ln2,\log(\lambda g(t)-1)-\epsilon(t)\}$.
\end{lem}
\begin{proof}
See Appendix \ref{Proof_lem_unique_root}
\end{proof}
Denote $\hat{c}(t)$ as the solution to $\vartheta(x;g(t),\epsilon(t))=0$
for $x\geq0$. From Lemma 2, the mapping from $g(t)$ and $\epsilon(t)$
to $\hat{c}(t)$ is unique. In addition, denote $\hat{\mu}(t)\triangleq\ln2\cdot2^{\hat{c}(t)+\epsilon(t)}/g(t)$,
$\bar{\mu}(t)\triangleq\ln2\cdot(1/g(t)+\bar{p}(t))$, and $\tilde{\mu}(t)\triangleq(\bar{p}(t)+\lambda)/(\log(1+\bar{p}(t)g(t))-\epsilon(t))$.
We have the following result to characterize the solution to $\mathscr{P}3$
under $N=1$.
\begin{prop}
\label{prop:op_single}(Optimal solution under $N=1$) The optimal
solution $(\bm{\phi},\bm{l})$ to $\mathscr{P}3$ under $N=1$ is
given by

(i) for\textup{ $t\in\mathcal{T}_{1}\triangleq\left\{ t\in\mathcal{T}:\hat{c}(t)\le\bar{c}(t)\right\} $}
\begin{equation}
\phi\left(t\right)=\begin{cases}
0 & \mu<\hat{\mu}\left(t\right)\\
\hat{c}\left(t\right)\tilde{l}\left(t\right) & \mu=\hat{\mu}\left(t\right)\\
\log\left(\mu g\left(t\right)/\ln2\right)-\epsilon\left(t\right) & \hat{\mu}\left(t\right)<\mu<\bar{\mu}\left(t\right)\\
\bar{c}\left(t\right) & \mu\ge\bar{\mu}\left(t\right)
\end{cases}\label{eq:ct_1}
\end{equation}
\begin{equation}
l\left(t\right)=\mathbb{I}\left\{ \mu>\hat{\mu}\left(t\right)\right\} +\tilde{l}\left(t\right)\mathbb{I}\left\{ \mu=\hat{\mu}\left(t\right)\right\} \label{eq:lt_1}
\end{equation}

(ii) for \textup{$t\in\mathcal{T}_{2}\triangleq\left\{ t\in\mathcal{T}:\hat{c}(t)>\bar{c}(t)\right\} $}
\begin{equation}
\phi\left(t\right)=\begin{cases}
0 & \mu<\tilde{\mu}\left(t\right)\\
\bar{c}\left(t\right)\tilde{l}\left(t\right) & \mu=\tilde{\mu}\left(t\right)\\
\bar{c}\left(t\right) & \mu>\tilde{\mu}\left(t\right)
\end{cases}\label{eq:ct_2}
\end{equation}
\begin{equation}
l\left(t\right)=\mathbb{I}\left\{ \mu>\tilde{\mu}\left(t\right)\right\} +\tilde{l}\left(t\right)\mathbb{I}\left\{ \mu=\tilde{\mu}\left(t\right)\right\} \label{eq:lt_2}
\end{equation}
where $\{\mu\ge0,\tilde{\boldsymbol{l}}\triangleq[\tilde{l}(t)]_{t\in\mathcal{T}}\}$
are parameters chosen to satisfy 
\begin{equation}
\tilde{\Upsilon}(\mu,\tilde{\boldsymbol{l}})\triangleq\sum_{t\in\mathcal{T}}\phi\left(t\right)=S\label{eq:dual-obj}
\end{equation}
and $\tilde{l}(t)\in[0,1],\,\forall t\in\mathcal{T}$.
\end{prop}
\begin{proof}
See Appendix \ref{sec:proof_op_single}.
\end{proof}
As a result, we can obtain the optimal solution to $\mathscr{P}3$
by solving equation (\ref{eq:dual-obj}). In addition, the throughput
function $\tilde{\Upsilon}(\mu,\tilde{\boldsymbol{l}})$ is proven
to be monotonically non-decreasing over $\mu$ and $\tilde{\boldsymbol{l}}$,
as shown in the following proposition.
\begin{prop}
\label{prop:monotonicity}(Monotonicity) For any $\mu_{1}>\mu_{2}\ge0$
and $\boldsymbol{1}\succcurlyeq\tilde{\boldsymbol{l}}_{1},\tilde{\boldsymbol{l}}_{2}\succcurlyeq\boldsymbol{0}$,\footnote{Here, the notation $\boldsymbol{a}\succcurlyeq\boldsymbol{b}$ means
that $\forall i$, $a_{i}\ge b_{i}$. The notation $\boldsymbol{a}\succ\boldsymbol{b}$
means that $\forall i$, $a_{i}\ge b_{i}$, and $\exists i,a_{i}>b_{i}$.} it holds that $\tilde{\Upsilon}(\mu_{1},\tilde{\boldsymbol{l}})\ge\tilde{\Upsilon}(\mu_{2},\tilde{\boldsymbol{l}})$.
For any $\boldsymbol{1}\succcurlyeq\tilde{\boldsymbol{l}}_{1}\succ\tilde{\boldsymbol{l}}_{2}\succcurlyeq\boldsymbol{0}$
and $\mu\ge0$, it holds that $\tilde{\Upsilon}(\mu,\tilde{\boldsymbol{l}}_{1})\ge\tilde{\Upsilon}(\mu,\tilde{\boldsymbol{l}}_{2})$.
\end{prop}
\begin{proof}
See Appendix \ref{sec:proof_prop_monotonicity}.
\end{proof}
\begin{figure}
\begin{centering}
\includegraphics[width=0.7\columnwidth]{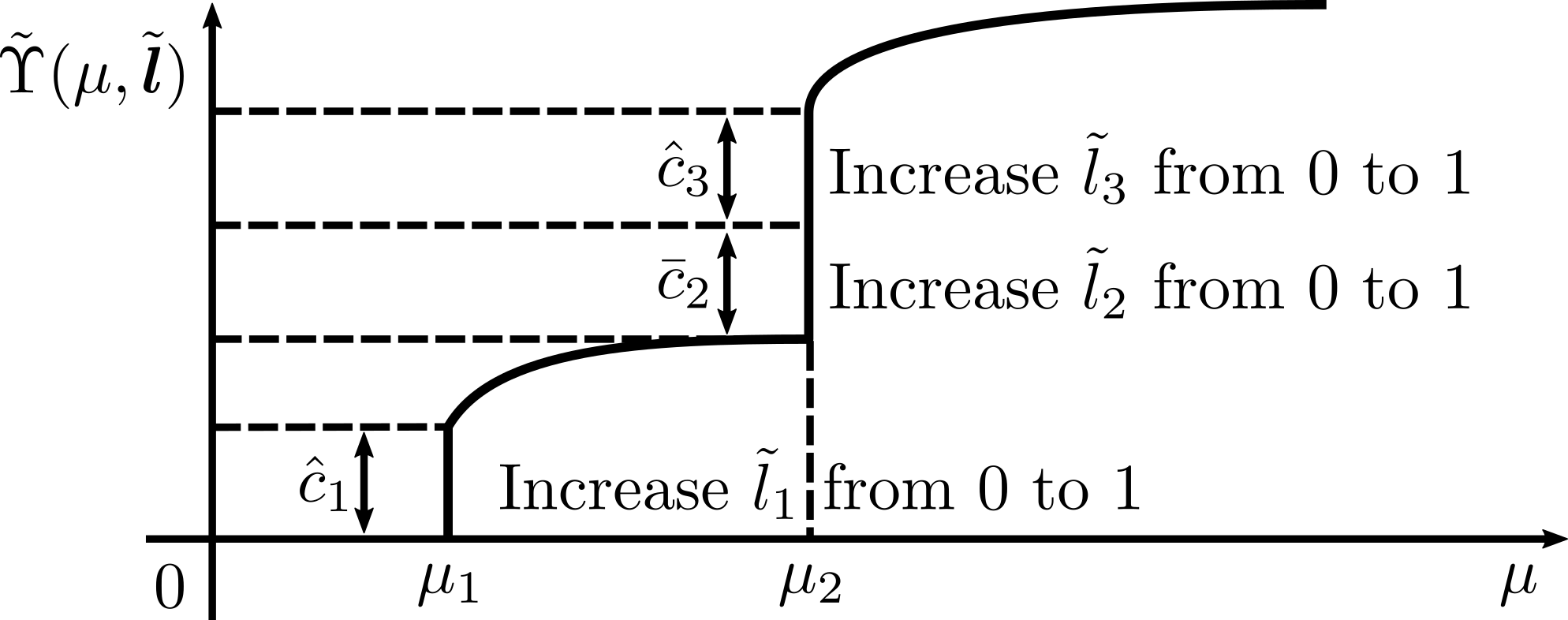}
\par\end{centering}
\caption{\label{fig:ill_monotonicity}Illustration of $\tilde{\Upsilon}(\mu,\tilde{\boldsymbol{l}})$
over $\mu$ and $\tilde{\boldsymbol{l}}$, where one slot is activated
when $\mu=\mu_{1}$, and two slots are activated when $\mu=\mu_{2}$.}
\end{figure}
Fig. \ref{fig:ill_monotonicity} illustrates monotonicity of $\tilde{\Upsilon}(\mu,\tilde{\boldsymbol{l}})$.
It is shown that (i) When $\mu<\mu_{1}$, $\tilde{\Upsilon}=0$; (ii)
When $\mu>\mu_{1}$, $\tilde{\Upsilon}$ is increasing over $\mu$;
(iii) When $\mu=\mu_{i}$, $\tilde{\Upsilon}$ is increasing over
$\tilde{l}_{i}$. Based on the monotonicity, one solution $\mu^{*}$
and $\tilde{\boldsymbol{l}}^{*}$ can be found using bisection search
as described in Algorithm \ref{alg:Fast-algorithm-n1}.

\begin{algorithm}
\# Bisection search to determine $\mu^{*}$.
\begin{enumerate}
\item Set $\mu_{\max}=\max_{t\in\mathcal{T}}\{\bar{\mu}(t),\tilde{\mu}(t)\}$,
$\mu_{\min}=0$;
\item Set $\mu=(\mu_{\max}+\mu_{\min})/2$, if $\tilde{\Upsilon}(\mu,\boldsymbol{0})>S$,
set $\mu_{\max}=\mu$; otherwise, set $\mu_{\min}=\mu$;
\item Repeat step 2) until $|\mu_{\max}-\mu_{\min}|$ is small enough;
\item Set $\mu^{*}=\mu_{\min}$.
\end{enumerate}
\# Bisection search to determine $\tilde{\boldsymbol{l}}^{*}$.
\begin{enumerate}
\item[5)]  Set $l_{\max}=1$, $l_{\min}=0$;
\item[6)] Set $l=(l_{\max}+l_{\min})/2$ and $\tilde{\boldsymbol{l}}=l\cdot\boldsymbol{1}$,
if $\tilde{\Upsilon}(\mu^{*},\tilde{\boldsymbol{l}})>S$, set $l_{\max}=l$;
otherwise, set $l_{\min}=l$;
\item[7)] Repeat step 2) until $|l_{\max}-l_{\min}|$ is small enough;
\item[8)] Set $l^{*}=l_{\max}$ and $\tilde{\boldsymbol{l}}=l^{*}\cdot\boldsymbol{1}$.
\end{enumerate}
\# Set $\boldsymbol{\phi}^{*}$ and $\boldsymbol{l}^{*}$ according
to (\ref{eq:ct_1}) \textendash{} (\ref{eq:lt_2}).

\caption{\label{alg:Fast-algorithm-n1}Fast dual-based algorithm for $N=1$.}
\end{algorithm}

The computational complexity of Algorithm \ref{alg:Fast-algorithm-n1}
is $\mathcal{O}(T)$ consisted of $\mathcal{O}(1)$ steps for the
bisection search of $\mu^{*}$ and $\tilde{\boldsymbol{l}}^{*}$,
where each step requires $\mathcal{O}(T)$ for calculating $\tilde{\Upsilon}(\mu,\tilde{\boldsymbol{l}})$.

\subsection{Efficient algorithm for $N>1$}

It is observed from $\mathscr{P}3$ that for $N>1$, the variables
are coupled over $n$ only by the constraint $\bm{l}\in\mathcal{L}$.
This observation suggests that $\mathscr{P}3$ can be decomposed into
$N$ subproblems, provided the coupled constraint involving $l_{n}(t)$
can be handled. To this end, a partial-dual-based algorithm is proposed,
where the partial dual problem is given by
\[
\underset{\boldsymbol{v}\succeq0}{\text{maximize}}\quad q\left(\boldsymbol{v}\right)
\]
where $q\left(\boldsymbol{\boldsymbol{v}}\right)$ is the partial
dual function over the Lagrangian parameter $\boldsymbol{\boldsymbol{v}}=[v(t)]_{t\in\mathcal{T}}$,
as given by the value of the following problem
\begin{align*}
\text{\ensuremath{\mathscr{P}4}}:\quad\underset{\bm{\Phi},\bm{L}}{\text{minimize}} & \quad Q\left(\bm{\Phi},\bm{L};\boldsymbol{v}\right)\\
\text{subject to} & \quad\sum_{t\in\mathcal{T}}\phi_{n}\left(t\right)\ge S_{n},\forall n\\
 & \quad\phi_{n}\left(t\right)\in\left[-\epsilon_{n}\left(t\right)l_{n}\left(t\right),\bar{c}_{n}\left(t\right)l_{n}\left(t\right)\right],\forall n,t\\
 & \quad l_{n}\left(t\right)\in\left[0,1\right],\forall n,t
\end{align*}
and $Q(\bm{\Phi},\bm{L};\boldsymbol{v})\triangleq\sum_{t\in\mathcal{T}}\sum_{n\in\mathcal{N}}((2^{\phi_{n}(t)/l_{n}(t)+\epsilon_{n}(t)}-1)/g_{n}(t)+\lambda)l_{n}(t)+\sum_{t\in\mathcal{T}}v(t)(\sum_{n\in\mathcal{N}}l_{n}(t)-1)$.

Denote $\lambda\left(t\right)\triangleq\lambda+v\left(t\right)$,
the objective function can be expressed as $Q(\bm{\Phi},\bm{L};\boldsymbol{v})=\sum_{n\in\mathcal{N}}\sum_{t\in\mathcal{T}}((2^{\phi_{n}(t)/l_{n}(t)+\epsilon_{n}(t)}-1)/g_{n}(t)+\lambda(t))l_{n}(t)-\sum_{t\in\mathcal{T}}v\left(t\right)$.
Then the problem $\text{\ensuremath{\mathscr{P}4}}$ can be decoupled
to $N$ following sub-problem
\begin{align*}
\underset{\bm{\phi}_{n},\bm{l}_{n}}{\text{minimize}} & \quad\sum_{t\in\mathcal{T}}\left(\frac{2^{\frac{\phi_{n}\left(t\right)}{l_{n}\left(t\right)}+\epsilon_{n}\left(t\right)}-1}{g_{n}\left(t\right)}+\lambda\left(t\right)\right)l_{n}\left(t\right)-\sum_{t\in\mathcal{T}}\frac{v\left(t\right)}{N}\\
\text{subject to} & \quad\sum_{t\in\mathcal{T}}\phi_{n}\left(t\right)\ge S_{n}\\
 & \quad\phi_{n}\left(t\right)\in\left[-\epsilon_{n}\left(t\right)l_{n}\left(t\right),\bar{c}_{n}\left(t\right)l_{n}\left(t\right)\right],\forall t\\
 & \quad l_{n}\left(t\right)\in\left[0,1\right],\forall t.
\end{align*}

The above problem is equivalent to $\mathscr{P}3$ for $N=1$ case
discussed in Section \ref{sec:Dual_algorithm}, where the solution
is given in Proposition \ref{prop:op_single} with the constant penalty
$\lambda$ to be replaced by $\lambda\left(t\right)$, and a constant
$-\sum_{t\in\mathcal{T}}v\left(t\right)/N$ to be added into the objective.

For maximizing $q\left(\boldsymbol{v}\right)$, the gradient projection
method is used, and the algorithm stops when the $||\boldsymbol{v}_{k+1}-\boldsymbol{v}_{k}||_{F}$
is small enough. The detailed algorithm is shown in Algorithm \ref{alg:alternative_opt}.

\begin{algorithm}
\# Initialization process: set $\boldsymbol{v}_{0}$.
\begin{enumerate}
\item Alternatively optimize $(\boldsymbol{\phi}_{n},\boldsymbol{l}_{n})$
from $n=1$ to $N$ by Algorithm \ref{alg:Fast-algorithm-n1};
\item Update $\boldsymbol{v}_{k+1}$ using gradient projection method and
go to step 1 until $||\boldsymbol{v}_{k+1}-\boldsymbol{v}_{k}||_{F}$
is small enough.
\end{enumerate}
\# Output $(\boldsymbol{\Phi}^{(k)},\boldsymbol{L}^{(k)})$.

\caption{\label{alg:alternative_opt}Alternative optimization algorithm for
$N>1$.}
\end{algorithm}

Algorithm \ref{alg:alternative_opt} converges because the dual problem
is convex \cite{boyd2004convex}. Moreover, the converged values $(\boldsymbol{\Phi}^{(k)},\boldsymbol{L}^{(k)})$
is the optimal solution to $\mathscr{P}3$ because Slater\textquoteright s
condition holds and strong duality holds since the objective function
in $\mathscr{P}3$ is convex and all the constraints in $\mathscr{P}3$
are affine \cite{boyd2004convex}.

The computational complexity of Algorithm \ref{alg:alternative_opt}
is $\mathcal{O}(NT)$ per iteration, constructed by $N$ times of
$\mathcal{O}(T)$ for calculating $(\boldsymbol{\phi}_{n},\boldsymbol{l}_{n})$.

\section{Rounding Strategy\label{sec:Rounding-Policy}}

Although we can efficiently solve for $\mathscr{P}3$ via Algorithm
2, which equivalently solves for $\mathscr{P}2$, the solution may
not be optimal to $\mathscr{P}1$, since $\mathscr{P}2$ is only a
relaxed version of $\mathscr{P}1$. In some extreme cases, the solution
to $\mathscr{P}2$ can lead to an unexpectedly large gap from the
optimal value of $\mathscr{P}1$. In this section, we develop a cost-aware
rounding strategy to adjust the solution obtained from $\mathscr{P}2$,
and show that the solution with rounding strategy is asymptotically
optimal in small slot duration.

\subsection{Challenge due to non-unique solutions to $\mathscr{P}2$}

One particular scenario where the solution from problem $\mathscr{P}2$
yields a significant performance gap from $\mathscr{P}1$ is when
the channel $g_{n}(t)$ is constant over $t$. In this case, the optimal
solutions to $\mathscr{P}2$ are not unique, and a possible solution
is to allocate equal power to a number of time slots which can be
fractionally used, i.e., $\sum_{n\in\mathcal{N}}l_{n}(t)$ is strictly
less than $1$. This phenomenon can be understood by the water-filling
solution for capacity-achieving power allocation, where the time slots
with the same channel gain will be allocated with the same power and
the same frequency resource $l_{n}(t)$. While this solution is optimal
in minimizing the relaxed cost function $\tilde{F}$ in (\ref{eq:erg_relax}),
it is not optimal for the original cost function in (\ref{eq:erg_def}),
due to the discontinuous indicator function, $\ensuremath{\mathbb{I}\{\sum_{n\in\mathcal{N}}l_{n}(t)>0\}}$.

We construct a numerical example in Appendix \ref{sec:proof-of-gap-23}
and show that directly applying the solution obtained in $\mathscr{P}2$
may yield an additional cost of $(T-1)\lambda$ in (\ref{eq:erg_def})
above the minimum. Note that such a performance gap can be arbitrarily
large for an arbitrarily large duration $T$.

Here, we investigate the conditions that characterize the set of equivalent
solutions to$\mathscr{P}2$. Given a transmission strategy $(\bm{P},\bm{L})$,
let $\bar{\mathcal{S}}(\boldsymbol{L})$ denote the set of partially
used slots defined as $\{t\in\mathcal{T}:\sum_{n\in\mathcal{N}}l_{n}(t)\notin\{0,1\}\}$.
Let $\mathcal{T}_{n}(\boldsymbol{L})$ denote the set of slots allocated
to node $n$ as $\{t\in\mathcal{T}:l_{n}(t)>0\}$. Then, the set of
slots that are partially used and specifically used by node $n$ is
given by $\bar{\mathcal{S}}_{n}(\boldsymbol{L})\triangleq\bar{\mathcal{S}}(\boldsymbol{L})\cap\mathcal{T}_{n}(\boldsymbol{L})$.
\begin{prop}
\label{prop:opt_preserve_c}(Equivalent solutions to $\mathscr{P}2$)
For any optimal solution to $\mathscr{P}2$, $(\bm{P}^{*},\bm{L}^{*})$,
a feasible solution $(\bm{P},\bm{L})$ is also optimal if the following
conditions are satisfied: (i) $p_{n}(t)=p_{n}^{*}(t),\,\forall n\in\mathcal{N},t\in\mathcal{T}$;
(ii) $l_{n}(t)=l_{n}^{*}(t),\,\forall t\notin\bar{\mathcal{S}}_{n}(\boldsymbol{L}^{*})$,
$\forall n\in\mathcal{N}$; (iii) $\sum_{t\in\bar{\mathcal{S}}_{n}(\boldsymbol{L}^{*})}(p_{n}(t)+\lambda)l_{n}(t)=\sum_{t\in\bar{\mathcal{S}}_{n}(\boldsymbol{L}^{*})}(p_{n}(t)+\lambda)l_{n}^{*}(t)$,
$\forall n\in\mathcal{N}$.
\end{prop}
\begin{proof}
See Appendix \ref{sec:proof-lem-opt_preserve_c}.
\end{proof}

\begin{figure}
\begin{centering}
\includegraphics{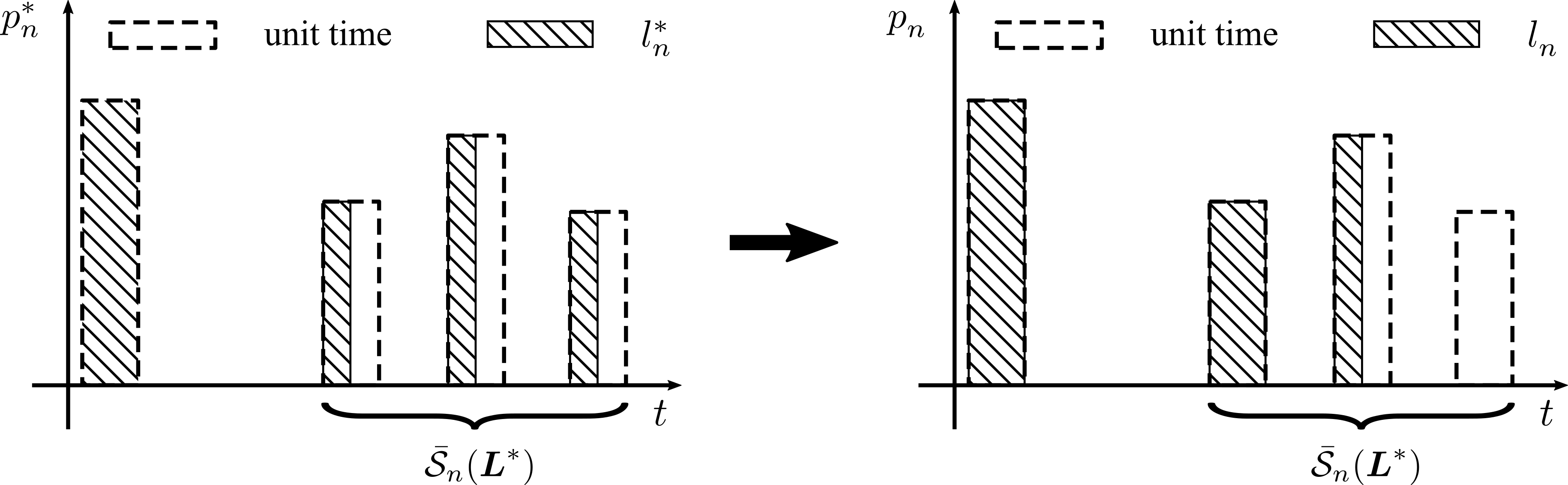}
\par\end{centering}
\caption{\label{fig:ill_opt_pre_con}The equivalence of multiple solutions
under the optimality preservation conditions for receiver $n$.}
\end{figure}

As a result, there are multiple optimal solutions to $\mathscr{P}2$
if the cardinality of $\bar{\mathcal{S}}_{n}(\boldsymbol{L}^{*})$
is greater than 1. Fig. \ref{fig:ill_opt_pre_con} illustrates the
equivalence of multiple solutions under the conditions in Proposition
\ref{prop:opt_preserve_c}. The solutions in the two subfigures of
Fig. \ref{fig:ill_opt_pre_con} have identical costs in (\ref{eq:erg_relax}),
because the power allocation in the right subfigure is to concentrate
the two half slots with the same power allocation in the left subfigure
into one full slot. In addition, both solutions have identical throughput
according to Proposition \ref{prop:opt_preserve_c}, because the cost
in (\ref{eq:erg_relax}) consumed by the slots in set $\bar{\mathcal{S}}_{n}(\boldsymbol{L}^{*})$
of the two solutions are identical. However, although the two solutions
yield the same cost in (\ref{eq:erg_relax}), the solution on the
right is more desired because it uses only 3 time slots as opposed
to 4 time slots used by the solution on the left.

\subsection{Cost-aware rounding algorithm}

To circumvent the issue from the non-unique optimal solutions, we
develop an easy-to-implement strategy to find a better solution. At
first, we prove that the cost gap is due to the number of partially
used slots.
\begin{prop}
\label{prop:perf_cap_non_r}(Cost upper bound) Denote $F^{*}$ as
the minimum cost of $\mathscr{P}1$. For any optimal solution $(\bm{P}^{*},\bm{L}^{*})$
to $\mathscr{P}2$, the cost in (\ref{eq:erg_def}) satisfies $F(\bm{P}^{*},\bm{L}^{*})-F^{*}\le|\bar{\mathcal{S}}(\boldsymbol{L}^{*})|\lambda$.
\end{prop}
\begin{proof}
See Appendix \ref{sec:proof-prop-perf-cap-non-r}.
\end{proof}
As a result, the performance gap to the optimum of $\mathscr{P}1$
can be decreased by reducing the number of the partially-used slot
sets, $|\bar{\mathcal{S}}(\boldsymbol{L})|$.

Based on Proposition \ref{prop:opt_preserve_c}, one can derive a
set of optimal solutions $\{(\bm{P},\bm{L})\}$ based on a particular
solution $(\bm{P}^{*},\bm{L}^{*})$ to $\mathscr{P}2$. In addition,
we are interested in the solution that minimizes $|\bar{\mathcal{S}}(\boldsymbol{L})|$
because of Proposition \ref{prop:perf_cap_non_r}. Given a particular
solution $(\bm{P}^{*},\bm{L}^{*})$ to $\mathscr{P}2$, compute the
parameter $\hat{\mathcal{S}}_{n}=\bar{\mathcal{S}}_{n}(\boldsymbol{L}^{*})$
and $\hat{F}_{n}=\sum_{t\in\hat{\mathcal{S}}_{n}}(p_{n}^{*}(t)+\lambda)l_{n}^{*}(t).$
An equivalent solution to $\mathscr{P}2$ with a minimum $|\bar{\mathcal{S}}(\boldsymbol{L})|$
can be found by solving the following problem:
\begin{align}
\mathscr{P}5:\quad\underset{\bm{L}}{\text{minimize}} & \quad\left|\bar{\mathcal{S}}\left(\boldsymbol{L}\right)\right|\nonumber \\
\text{subject to} & \quad\boldsymbol{l}\left(t\right)\in\mathcal{L},\forall t\in\hat{\mathcal{S}}_{n}\label{eq:car_min_c1}\\
 & \quad l_{n}\left(t\right)=l_{n}^{*}\left(t\right),\forall t\notin\hat{\mathcal{S}}_{n},\forall n\label{eq:car_min_c2}\\
 & \quad\sum_{t\in\hat{\mathcal{S}}_{n}}\left(p_{n}^{*}\left(t\right)+\lambda\right)l_{n}\left(t\right)=\hat{F}_{n},\forall n\label{eq:car_min_c3}
\end{align}
and $\bm{P}=\bm{P}^{*}$. It follows that $p_{n}(t)=p_{n}^{*}(t)$
meets the first condition in Proposition \ref{prop:opt_preserve_c},
and constraint (\ref{eq:car_min_c2}) and constraint (\ref{eq:car_min_c3})
correspond to the last two conditions in Proposition \ref{prop:opt_preserve_c},
respectively. Condition (\ref{eq:car_min_c1}) ensures the solution
is feasible to $\mathscr{P}2$. Thus, the optimal solution to $\mathscr{P}5$
is also an optimal solution to $\mathscr{P}2$ according to Proposition
\ref{prop:opt_preserve_c}.

Problem$\mathscr{P}5$ can be solved by an sequential approach by
decomposing it into $N$ subproblems by partially optimizing variable
$l_{n}(t)$. Specially, given $l_{m}^{*}(t)$, $m\neq n$, the constraint
$l_{n}(t)\in\mathcal{L}$ in $\mathscr{P}5$ becomes $0\leq l_{n}(t)\leq1-\sum_{m\neq n}l_{m}^{*}(t)$,
and the objective becomes $|\bar{\mathcal{S}}_{n}(\boldsymbol{L})|$.
Then the $n$th subproblem of $\mathscr{P}5$ degenerates to minimizing
the cardinality of the set $\bar{\mathcal{S}}_{n}\left(\boldsymbol{L}\right)$
with modified constraint $0\leq l_{n}(t)\leq1-\sum_{m\neq n}l_{m}^{*}(t)$.

For simplification, we use a bisection searching method to approximately
solve the subproblem, that is, updating $l_{n}\left(t\right)$ for
$t\in\hat{\mathcal{S}}_{n}$ while leaving $l_{n}\left(t\right)$
for $t\notin\hat{\mathcal{S}}_{n}$ unchanged, as shown in Algorithm
\ref{alg:Rounding-Algorithm}. By arbitrarily fixing the order of
elements in the set $\bar{\mathcal{S}}_{n}(\boldsymbol{L})$, as $\{q_{1},\cdots q_{K}\}$,
where $K=|\hat{\mathcal{S}}_{n}|$, we opt for the first $k$ slots
in $\hat{\mathcal{S}}_{n}$ for full usage, and allocate the $(k+1$)th
slot according to the constraint (\ref{eq:car_min_c3}), while leaving
the remaining $K-k-1$ slots unused. In other words,
\begin{equation}
l_{n}\left(t\right)=\begin{cases}
1-\sum_{m\neq n}l_{m}^{*}(t) & t\in\{q_{1},\cdots q_{k}\}\\
\tilde{l}_{n}\left(t\right) & t=q_{k+1}\\
0 & t\in\{q_{k+1},\cdots q_{K}\}
\end{cases}\label{eq:allocation_lnt}
\end{equation}
where $k\in\mathbb{Z}^{+}$ and $\tilde{l}_{n}(t)\in[0,1-\sum_{m\neq n}l_{m}^{*}(t))$
is chosen to satisfy condition (\ref{eq:car_min_c3}), that is the
solution of the following equation
\begin{align}
 & \sum_{n\in\{q_{1},\cdots,q_{k}\}}\left(p_{n}^{*}\left(t\right)+\lambda\right)\left(1-\sum_{m\neq n}l_{m}^{*}\left(t\right)\right)\nonumber \\
 & \quad\quad\quad\quad\quad+\left[\left(p_{n}^{*}\left(t\right)+\lambda\right)\tilde{l}_{n}\left(t\right)\right]_{n=q_{k+1}}=\hat{F}_{n}\label{eq:bi-ss-rounding}
\end{align}
which can be solved by bisection searching method because the summation
function is monotonically increasing over $k$.

\begin{algorithm}
\# Initialization process: set $n\leftarrow0$.
\begin{enumerate}
\item[1)] Calculate $\hat{\mathcal{S}}_{n}$ and $\hat{F}_{n}$ based on $(\bm{P}^{*},\bm{L}^{*})$;
\item[2)] Arbitrarily fixing the order of elements in the set $\hat{\mathcal{S}}_{n}$,
and bisection search $k$ and $\tilde{l}_{n}(t)$ according to (\ref{eq:bi-ss-rounding});
\item[3)] Update $l_{n}^{*}(t)$ for $t\in\hat{\mathcal{S}}_{n}$ according
to (\ref{eq:allocation_lnt}), set $n\leftarrow n+1$, and go to step
1 until $n>N$.
\end{enumerate}
\# Output $\hat{\boldsymbol{L}}=\boldsymbol{L}^{*}$.

\caption{\label{alg:Rounding-Algorithm}Cost-ware Rounding Algorithm}
\end{algorithm}

\subsubsection{Complexity analysis}

The Algorithm \ref{alg:Rounding-Algorithm} consists of $N$ rounds
of bisection searching for $k$ and $\tilde{l}_{n}(t)$ with $\mathcal{O}(T\log(T))$
computational complexity, consisted of $\mathcal{O}(\log(T))$ steps
for the bisection search of $k^{*}$ and $\tilde{l}_{n}^{*}(t)$,
where each step requires $\mathcal{O}(T)$ for calculating $\sum_{n\in\{q_{1},\cdots,q_{k}\}}(p_{n}^{*}(t)+\lambda)(1-\sum_{m\neq n}l_{m}^{*}(t))$.
As a result, the computational complexity of the Algorithm \ref{alg:Rounding-Algorithm}
is $\mathcal{O}(NT\log(T))$.

\subsubsection{Performance analysis}

Denote the slot duration as $\delta$, and the transmission duration
as $\Gamma$, and consequently, we have $T=\Gamma/\delta$. We prove
that the cost gap reduces to $N\delta\lambda$ by cost-aware rounding
algorithm.
\begin{prop}
(Asymptotic optimality\label{prop:alg-opt}) Let $F^{*}$ be the minimum
cost of $\mathscr{P}1$ and $\hat{\boldsymbol{L}}$ be the output
of Algorithm \ref{alg:Rounding-Algorithm}. Then, $F(\boldsymbol{P}^{*},\hat{\boldsymbol{L}})-F^{*}\le N\delta\lambda$
, {\em i.e.}, $F(\boldsymbol{P}^{*},\hat{\boldsymbol{L}})\to F^{*}$
as $\delta\to0$.
\end{prop}
\begin{proof}
See Appendix \ref{sec:proof_gap_upper_bound}.
\end{proof}
It is shown from Proposition \ref{prop:alg-opt} that the proposed
rounding algorithm yields a solution that is asymptotically optimal
for small slot duration $\delta$.

\begin{figure}
\begin{centering}
\includegraphics[width=1\columnwidth]{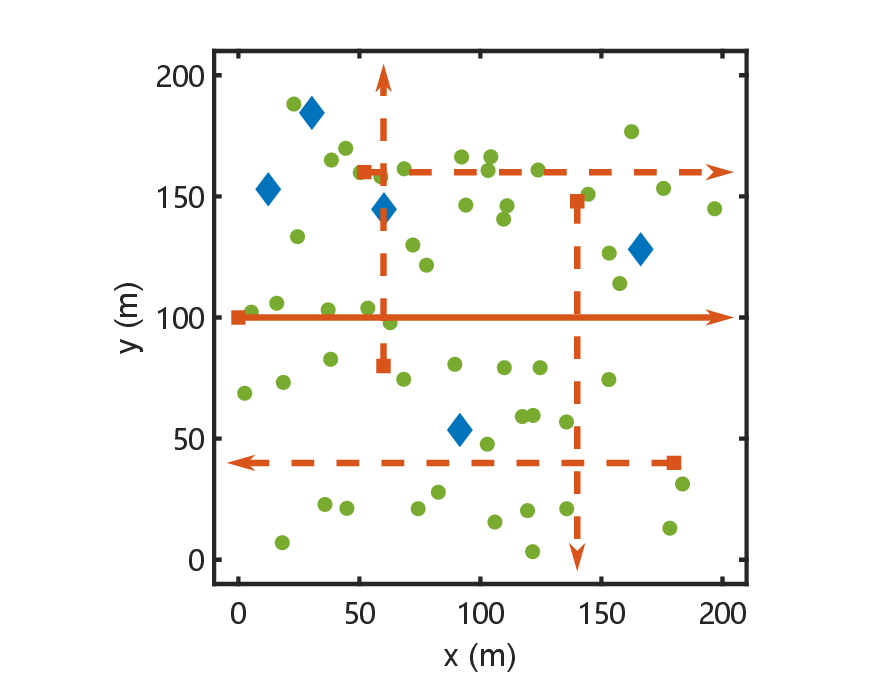}
\par\end{centering}
\caption{\label{fig:uav_layout}Illustration for \acpl{bs} positions (blue
diamonds), user positions (green circles), and \ac{uav} routes (red
lines with directions), where the solid line represents the route
of transmitting \ac{uav} and the dashed lines represent the route
of receiving \ac{uav}.}
\end{figure}

\section{Simulation\label{sec:Simulation}}

Consider a $200\times200\text{ m\ensuremath{^{2}}}$ area, where $5$
\acpl{bs}, acting as the nodes with known locations, and $N_{\text{ue}}=100$
ground users, acting as the nodes with unknown locations, are randomly
distributed, one \ac{uav}, serving as the transmitting node, navigates
along the $y=100$ route, and four \acpl{uav}, two flying horizontally
and two vertically, acting as the receiving nodes, traverse this area,
as shown in Fig. \ref{fig:uav_layout}. The height of users, \acpl{bs},
vertical-route \acpl{uav}, and horizontal-route \acpl{uav} are set
0 m, 10 m, 95 m, and 100 m. The speeds of transmitting and receiving
\acpl{uav} are 5 m/s and 3 m/s, respectively.

The channels are realized by $h_{j}=g_{j}\xi_{j}$ according to (\ref{eq:channel_model}),
where $g_{j}$ includes path loss and shadowing $[-g_{j}]_{dB}=\text{PL}_{j}+\chi$,
and $\xi_{j}$ follows a Gamma distribution $\text{G}(\kappa_{j},1/\kappa_{j})$.
The shape parameter $\kappa_{j}$ of is chosen randomly from $1$
to $30$. The shadowing is modeled by a log-normal distribution, where
$\chi$ follows zero mean and $8$ variance with correlated distance
$5$ m. The propagation for \ac{los} and \ac{nlos} links is taken
from the 3GPP \ac{umi} model in \cite{TR36814} as
\[
\text{PL}_{j}=\begin{cases}
22.0+28.0\log_{10}(d_{j})+20\log_{10}(f_{c}), & \text{LOS link}\\
22.7+36.7\log_{10}(d_{j})+26\log_{10}(f_{c}) & \text{NLOS link}
\end{cases}
\]
where $d_{j}$ is the distance between the transmitting node and the
receiving or interfered nodes, and $f_{c}=3$ Ghz represents the carrier
frequency.

Since there is low blockage probability between \acpl{uav}, we assume
the channel between transmitting node to receiving nodes is always
in \ac{los}, {\em i.e.}, $\text{PL}_{j}=\text{PL}_{\text{LOS}}$,
and there is no shadowing between them, {\em i.e.}, $\chi=0$.

The blockage status between \ac{uav} and ground nodes is simulated
using the following steps. Initially, we employ the \ac{los} probability
model \cite{AlhKanLar:J14,MozSadBen:J17} 
\begin{equation}
\mathbb{P}\left(\text{LOS},\theta\right)=\frac{1}{1+a\exp\left(-b\left[\theta-a\right]\right)}\label{eq:def_p_los}
\end{equation}
to generate \ac{los} probability according to the elevation angle
$\theta=\sin^{-1}(u_{j}/d_{j})\times180/\pi$, where $u_{j}$ represents
their relative height, and the parameters are set $a=11.95$ and $b=0.14$
\cite{MozSadBen:J17}. Subsequently, the blockage status is determined
based on a threshold $p_{\text{th}}=0.5$. The channel is considered
blocked if $\mathbb{\mathbb{P}}\left(\text{LOS},\theta\right)<p_{\text{th}}$,
otherwise, the channel is in \ac{los}.

\begin{figure}
\begin{centering}
\includegraphics[width=1\columnwidth]{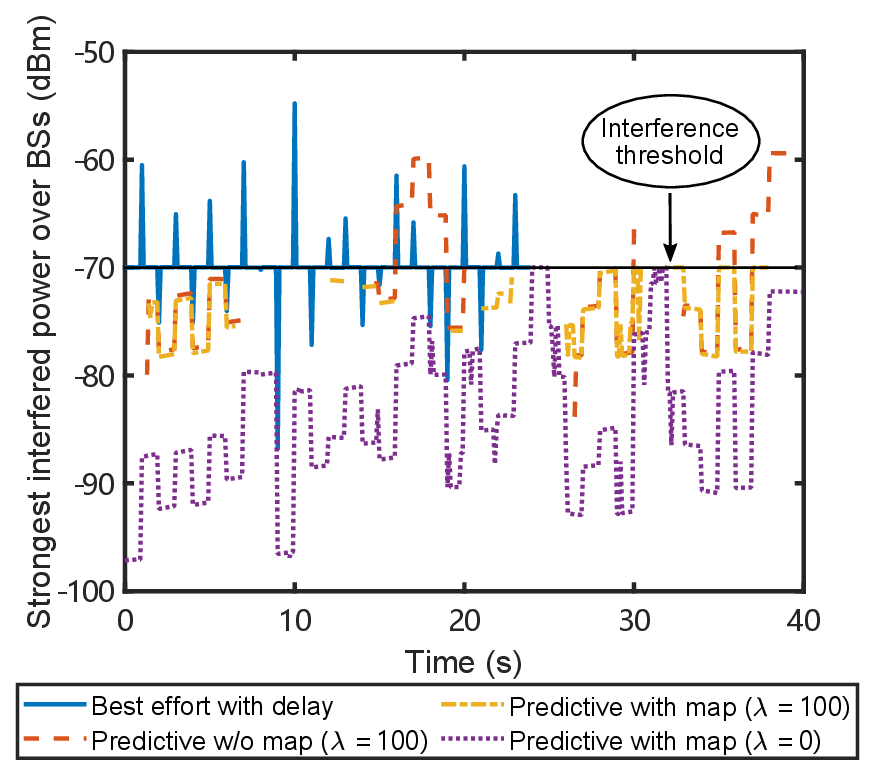}
\par\end{centering}
\caption{\label{fig:itf_leakage_BSs}Strongest interference power received
among the 5 \acpl{bs} at each time slot over a duration of 40 seconds.}
\end{figure}

\begin{figure}
\begin{centering}
\includegraphics[width=1\columnwidth]{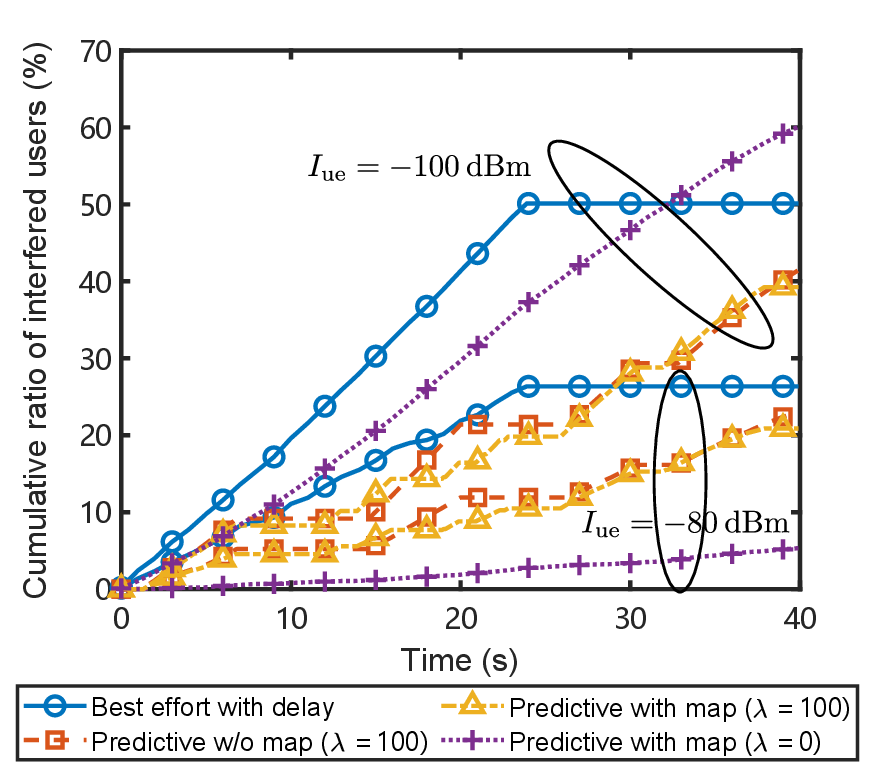}
\par\end{centering}
\caption{\label{fig:itf_leakage_users}Cumulative radio of interfered users
at each time slot determined based on different interference thresholds
$I_{\text{ue}}$.}
\end{figure}

We evaluate the following baseline schemes

\subsubsection{Best effort with delayed channel information (Best effort with delay)}

This scheme optimizes the transmission power based solely on the current
instantaneous channels and interference constraints, employing a best-effort
transmission strategy. Specifically, for each upcoming time slot,
it maximizes the throughput subject to the power and interference
constraints, based on the current instantaneous channel conditions.

\subsubsection{Predictive communication without a radio map (Predictive w/o map)}

When a radio map is not available, the probabilistic \ac{los} model
(\ref{eq:def_p_los}) is used for constructing an average channel
gain model as $\bar{g}_{j}=\mathbb{P}(\text{LOS},\theta)\cdot\text{PL}_{\text{LOS}}+(1-\mathbb{P}(\text{LOS},\theta))\cdot\text{PL}_{\text{NLOS}}$.
The predicted channel capacity is approximated by $\bar{c}_{n}\triangleq\text{\ensuremath{\log\left(1+\beta p_{n}g_{n}\right)}}$,
where $\beta=0.5$ to back-off for the small-scale fading. The scheme
then solves $\mathscr{P}1$ based on the predicted average channel
gain $\bar{g}_{j}$ and the predicted channel capacity $\bar{c}_{n}$.

\subsubsection{Proposed predictive optimization with radio maps (Predictive with
map)}

The proposed scheme first calculates the attenuation $\epsilon_{n}(t)$
according to Lemma \ref{lem:Lower_bound_of_thp_exp} and the interference-equivalent
power upper bound $\bar{p}(t)$ according to (\ref{eq:thp_c_v2}),
based on radio maps $\bm{\Theta}$. It then obtains an optimal solution
to $\mathscr{P}2$ using Algorithm \ref{alg:alternative_opt}. Finally,
the transmission strategy plan is obtained through rounding the solution
to $\mathscr{P}2$ based on Algorithm \ref{alg:Rounding-Algorithm}.

\subsection{Interference leakage}

We first evaluate the impact of interference leakage on \acpl{bs}
with the a required data size $S_{n}=500\text{ Mbits}$ and interference
threshold $I_{\text{bs}}=-70$ dBm. Fig.~\ref{fig:itf_leakage_BSs}
shows the strongest interference power received among the 5 \acpl{bs}
at each time slot over a duration of $T=40$ seconds. The results
demonstrate that the proposed predictive optimization, which explores
radio maps, can prevent all unintended receivers with known positions
from experiencing interference (beyond the interference threshold).
In contrast, when a radio map is not available, the baseline schemes
cannot guarantee the interference constraints to be satisfied.

Subsequently, we evaluate the impact of interference leakage on users
at unknown locations, as depicted in Fig.~\ref{fig:itf_leakage_users}.
For each time slot, a user is deemed to be interfered if the received
interference power exceeds the interference threshold of users $I_{\text{ue}}$.
Fig.~\ref{fig:itf_leakage_users} shows the cumulative ratio of interfered
users at each time slot over a period of $T=40$ seconds, where the
ratio of interfered users at time $t$ is computed as $r\left(t\right)=\sum_{q=1}^{N_{\text{ue}}}\mathbb{E}\left\{ \mathbb{I}\left\{ \boldsymbol{p}\left(t\right)\boldsymbol{l}^{H}\left(t\right)h_{q}\left(t\right)>I_{\text{ue}}\right\} \right\} /\left(N_{\text{ue}}T\right)$,
where $h_{q}\left(t\right)$ is the channel gain between transmitting
node to user $q$. It is observed that the cumulative ratio of interfered
users for the best effort scheme over the 40 seconds, statistically
roughly $50\%$ of users experience at least $-100\text{ dBm}$ interference
and $20\%$ of users experience at least $-80\text{ dBm}$ interference
for the entire period. For the proposed scheme with radio maps, these
ratios reduce to $40\%$ (at $\lambda=100$) and $5\%$ (at $\lambda=0$),
respectively, achieving a gain of 10\%\textendash 15\%. This is because
the proposed scheme can reduce interference power by employing radio
maps to concentrate the power on the good channels, and reduce the
interference time by introducing the weighting factor to concentrate
the transmission in fewer slots. The property of short transmission
duration is also observed in Fig.~\ref{fig:itf_leakage_BSs} for
the map-assisted scheme under $\lambda=100$.

Fig.~\ref{fig:itf_leakage_users_S} evaluates the impact of interference
leakage over the delivery of different data volume $S_{n}$, where
the average ratio of interfered users is computed as $\sum_{t\in\mathcal{T}}r\left(t\right)$.
On average, the number of users affected by interference is reduced
by over $10\%$, even in the case without a radio map. With the employment
of radio maps, this figure is further reduced by more than $14\%$
(at $I_{\text{ue}}=-100\text{ dBm}$) and $70\%$ (at $I_{\text{ue}}=-80\text{ dBm}$),
leading to fewer users being interfered with even without knowing
their locations or channels. The superiority is significant when handling
medium-sized data volumes, for example, when $S_{n}=400\text{ Mbits}$,
the proposed predictive scheme with radio maps outperforms the best
effort scheme by $23\%$ (at $I_{\text{ue}}=-100\text{ dBm}$) and
$90\%$ (at $I_{\text{ue}}=-80\text{ dBm}$).
\begin{figure}
\begin{centering}
\includegraphics[width=1\columnwidth]{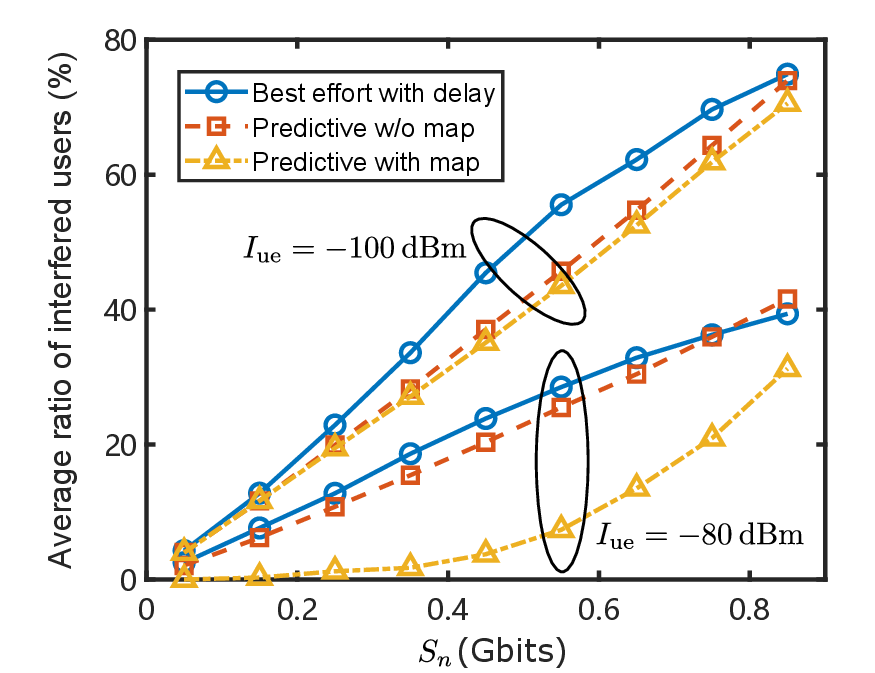}
\par\end{centering}
\caption{\label{fig:itf_leakage_users_S}Average ratio of interfered users
over different required data size determined by different interference
threshold $I_{\text{ue}}$.}
\end{figure}

\subsection{Convergence and complexity}

We evaluate the following algorithms that also consider the universal
interference and exploit the radio map, but use different strategies
to plan resources.

\subsubsection{LR \& SCA}

This algorithm solves the non-continuous non-convex problem $\mathscr{P}1$
through a four-step process. (i) Continuity Transformation: A new
parameter $a(t)\triangleq\sum_{n\in\mathcal{N}}l_{n}(t)$ is introduced,
replacing the indicator function $\mathbb{I}\{\sum_{n\in\mathcal{N}}l_{n}(t)>0\}$,
with constraints $a(t)\in[0,1]$ and $a(t)(1-a(t))\le0$; (ii) \ac{lr}
\cite{Fis:J81}: The non-convex constraint $a(t)(1-a(t))\le0$ is
incorporated into the objective function with a penalty factor; (iii)
\ac{sca} \cite{Raz:T14}: The local problem is approximated to a
convex problem using Taylor\textquoteright s approximation; (iv) CVX
solver \cite{MicSte:14}: The local convex problem is solved using
the CVX tool.

In addition, we evaluate the cost of the scheme without rounding,
{\em i.e.}, applying the solution to Algorithm \ref{alg:alternative_opt}
for transmission planning, the cost lower bound, {\em i.e.}, the
optimal value $\tilde{F}^{*}$ of the relaxed problem $\mathscr{P}2$,
and the theoretical largest performance gap of the proposed relax-then-round
strategy according to Proposition \ref{prop:perf_cap_non_r}, {\em i.e.},
$\tilde{F}^{*}+N\delta\lambda$.

Fig. \ref{fig:Cost_T} depicts the cost with the slot duration $\delta$.
It is observed that the cost of the proposed scheme with cost-aware
rounding (Algorithm \ref{alg:alternative_opt} \& \ref{alg:Rounding-Algorithm})
consistently remains below the theoretical upper bound (Proposition
\ref{prop:perf_cap_non_r}), and both converge towards the lower bound
as the slot duration $\delta$ diminishes. This validates the performance
gap solution in Proposition \ref{prop:perf_cap_non_r} and aligns
with the asymptotic optimality of the proposed scheme. In comparison,
a significant gap to the optimum is noted in the absence of a rounding
process (Algorithm \ref{alg:alternative_opt} only).

\begin{figure}
\begin{centering}
\includegraphics[width=1\columnwidth]{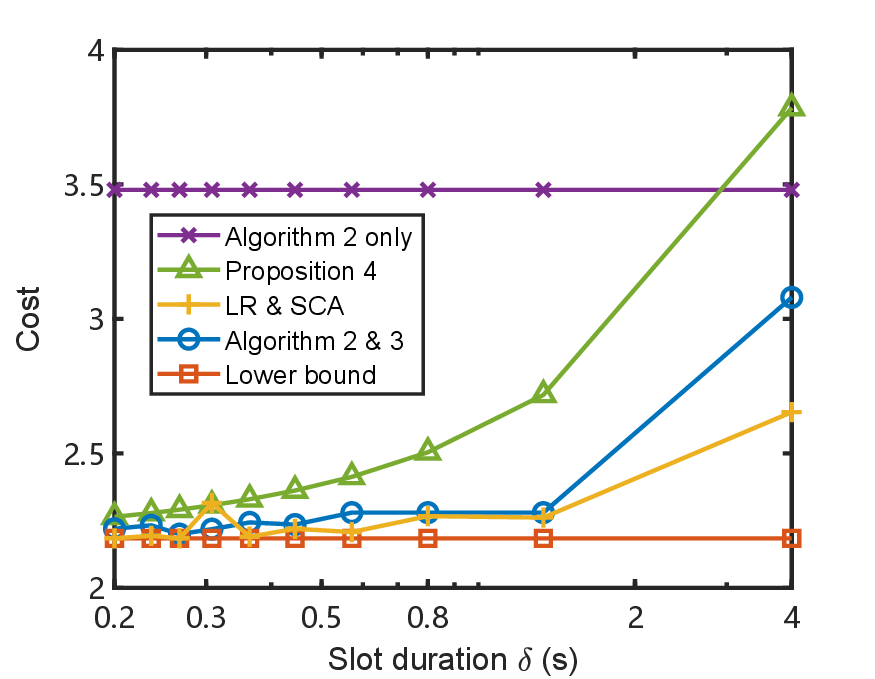}
\par\end{centering}
\caption{\label{fig:Cost_T}Costs over different slot durations $\delta$.}
\end{figure}

While \ac{lr} \& \ac{sca} and proposed schemes have similar performance,
the computational complexity is significantly different. Table. \ref{tab:Computation-time}
shows the running time over different slot lengths. The comparison
includes the \ac{lr} \& \ac{sca} and CVX \& Algorithm \ref{alg:Rounding-Algorithm}
schemes, with the latter addressing the proposed relax-then-rounding
problem not with the proposed dual-based algorithm solving problem
$\mathscr{P}3$, but with the CVX tool. It is demonstrated that the
proposed dual-based algorithm achieves a time-saving factor of 50
times less than the CVX scheme and 1000 times less than the \ac{lr}
\& \ac{sca} scheme for the case of $N=4$. These results suggest
that the proposed dual-based algorithm is significantly more efficient
than the general solver, and the proposed relax-then-round scheme
outperforms the general \ac{lr}-based and \ac{sca}-based algorithm
in terms of efficiency. In addition, for the case of $N=2$, the proposed
dual-based algorithm achieves a time-saving factor of 1000 times less
than the CVX scheme and 30,000 times less than the \ac{lr} \& \ac{sca}
scheme. This is because the parameter $\boldsymbol{v}$ in Algorithm
\ref{alg:alternative_opt} is converges faster for smaller $N$, such
that the running time is dominated by step 1 with computational complexity
$\mathcal{O}(NT)$.
\begin{center}
\begin{table}
\caption{\label{tab:Computation-time}Computational time (s) over different
slot number $T$ and receiving node $N$.}

\centering{}%
\begin{tabular}{c|ccccc}
\hline 
$T$ ($N=4$) & 10 & 50 & 90 & 130 & 170\tabularnewline
\hline 
\ac{lr} \& \ac{sca} & 106.4 & 357.9 & 587.5 & 1165 & 1222\tabularnewline
CVX \& Algorithm \ref{alg:Rounding-Algorithm} & 3.854 & 13.510 & 23.996 & 29.212 & 38.603\tabularnewline
Algorithm \ref{alg:alternative_opt} \& \ref{alg:Rounding-Algorithm} & 0.128 & 0.273 & 0.413 & 0.513 & 0.652\tabularnewline
\hline 
\hline 
$T$ ($N=2$) & 10 & 50 & 90 & 130 & 170\tabularnewline
\hline 
\ac{lr} \& \ac{sca} & 57.78 & 114.7 & 121.3 & 232.9 & 306.5\tabularnewline
CVX \& Algorithm \ref{alg:Rounding-Algorithm} & 2.306 & 6.099 & 10.26 & 14.18 & 18.32\tabularnewline
Algorithm \ref{alg:alternative_opt} \& \ref{alg:Rounding-Algorithm} & 0.002 & 0.003 & 0.004 & 0.008 & 0.009\tabularnewline
\hline 
\end{tabular}
\end{table}
\par\end{center}

\section{Conclusion\label{sec:Conclusion}}

This paper exploits radio maps and develops a relax-then-round scheme
for interference-aware predictive \ac{uav} communications. First,
a radio-map-assisted optimization problem is formulated, addressing
the uncertainty of future \ac{csi} and converting the interference
limit to an equivalent power limit by exploiting radio maps. Then,
the non-convex problem is relaxed and transformed into a convex one,
and a dual-based algorithm with complexity $\mathcal{O}(NT)$ per
iteration is proposed. Subsequently, a cost-aware rounding algorithm
with complexity $\mathcal{O}(NT\log(T))$ is developed, which is theoretically
proven to guarantee asymptotic optimality. Simulations validate that
the radio-map-assisted scheme effectively prevent all unintended receivers
with known positions from experiencing interference, and substantially
reduces the interference to the users at unknown locations. Furthermore,
the numerical results confirm that the proposed relax-then-round scheme
converges to the optimum as the slot duration diminishes, and the
running time of the proposed scheme achieves orders of magnitude of
reduction.

% ----
\appendices
% ----

\section{Proof of Lemma \ref{lem:Lower_bound_of_thp_exp}\label{sec:Proof_lem_lower_bound_of_thp_exp}}

Since $h_{n}$ follows Gamma distribution $\text{G}(\kappa_{n},g_{n}/\kappa_{n})$,
then $p_{n}h_{n}$ follows Gamma distribution $\text{G}(\kappa_{n},p_{n}g_{n}/\kappa_{n})$.
Based on logarithmic Jensen\textquoteright s gap theory \cite[Proposition 3]{SebKarDig:J18},
we have 
\[
\log\left(1+\mathbb{E}\left\{ p_{n}h_{n}\right\} \right)-\mathbb{E}\left\{ \log\left(1+p_{n}h_{n}\right)\right\} \le\epsilon_{n}
\]
where $\epsilon_{n}=\text{log}(e)/\kappa_{n}-\log(1+1/(2\kappa_{n})).$
Thus we have 
\begin{align*}
\mathbb{E}\left\{ \log\left(1+p_{n}h_{n}\right)\right\}  & \ge\log\left(1+\mathbb{E}\left\{ p_{n}h_{n}\right\} \right)-\epsilon_{n}\\
 & =\log\left(1+p_{n}g_{n}\right)-\epsilon_{n}.
\end{align*}

\section{Proof of Lemma \ref{lem:unique_root}\label{Proof_lem_unique_root}}

At first, the function $\vartheta(x)$ is monotonically decreasing
in $x>0$, because its first derivative $\frac{\partial\vartheta}{\partial x}=-x(\ln2)^{2}2^{x+\epsilon\left(t\right)}/g\left(t\right)$
is less than $0$ for $x>0$. Then, we have $\vartheta(0)=(2^{\epsilon(t)}-1)/g(t)+\lambda>0$
and 
\begin{align*}
 & \vartheta\left(\max\left\{ 2/\ln2,\log\left(\lambda g\left(t\right)-1\right)-\epsilon\left(t\right)\right\} \right)\\
 & \le\big(2^{x+\epsilon\left(t\right)}-1\big)/g\left(t\right)+\lambda-2\cdot\big(2^{x+\epsilon\left(t\right)}/g\left(t\right)\big)\\
 & =-2^{x+\epsilon\left(t\right)}/g\left(t\right)+\lambda-1/g\left(t\right)\\
 & \le-\left(\lambda g\left(t\right)-1\right)/g\left(t\right)+\lambda-1/g\left(t\right)\\
 & =0.
\end{align*}
As a result, the function $\vartheta(x)$ has a unique root in $x\ge0$,
and the root lies in $(0,\max\{2/\ln2,\log(\lambda g(t)-1)-\epsilon(t)\}]$.

\section{Proof of Proposition \ref{prop:op_single}\label{sec:proof_op_single}}

For any $t\in\mathcal{T}$, any feasible negative $\phi(t)$ will
decrease the throughput and increase the cost, thus there must be
optimal $\phi^{*}(t)\ge0$. In other words, the optimal solution to
problem $\mathscr{P}3$ under $N=1$ is equivalent to the solution
to the following problem
\begin{align*}
\mathscr{P}6:\quad\underset{\boldsymbol{\phi},\boldsymbol{l}}{\text{minimize}} & \quad\sum_{t\in\mathcal{T}}\left(\frac{2^{\frac{\phi\left(t\right)}{l\left(t\right)}+\epsilon\left(t\right)}-1}{g\left(t\right)}+\lambda\right)l\left(t\right)\\
\text{subject to} & \quad\sum_{t\in\mathcal{T}}\phi\left(t\right)\ge S\\
 & \quad\phi\left(t\right)\ensuremath{\in\left[0,\bar{c}\left(t\right)l\left(t\right)\right]},l\left(t\right)\in\left[0,1\right],\forall t.
\end{align*}
Since the objective function in $\mathscr{P}6$ is convex and all
the constraints in $\mathscr{P}6$ are affine, then Slater\textquoteright s
condition holds and strong duality holds \cite{boyd2004convex}. In
other words, the \ac{kkt} conditions provide necessary and sufficient
conditions for optimality \cite{boyd2004convex}.

\subsection{Optimality conditions}

Let $\boldsymbol{\Lambda}=\left[\mu,\boldsymbol{v}_{1},\boldsymbol{v}_{2},\boldsymbol{v}_{3},\boldsymbol{v}_{4}\right]$
be the Lagrangian parameter set, then the Lagrangian function of $\mathscr{P}6$
is given by
\begin{align*}
 & L\left(\boldsymbol{\phi},\boldsymbol{l},\bm{\Lambda}\right)\\
 & =\sum_{t\in\mathcal{T}}\left(\frac{2^{\frac{\phi\left(t\right)}{l\left(t\right)}+\epsilon\left(t\right)}-1}{g\left(t\right)}+\lambda\right)l\left(t\right)+\mu\left(S-\sum_{t\in\mathcal{T}}\phi\left(t\right)\right)\\
 & \quad+\sum_{t\in\mathcal{T}}\left(v_{1}\left(t\right)\left(-\phi\left(t\right)\right)+v_{2}\left(t\right)\left(\phi\left(t\right)-\bar{c}\left(t\right)l\left(t\right)\right)\right)\\
 & \quad+\sum_{t\in\mathcal{T}}\left(v_{3}\left(t\right)\left(-l\left(t\right)\right)+v_{4}\left(t\right)\left(l\left(t\right)-1\right)\right).
\end{align*}
Let $(\boldsymbol{\phi}^{*},\boldsymbol{l}^{*})$ be the optimal solution
to $\mathscr{P}6$ and $\boldsymbol{\Lambda}^{*}$ be the optimal
Lagrange multiplier to its dual problem. Then we can derive from the
\ac{kkt} optimality conditions that $\forall t$, $(\boldsymbol{\phi}^{*},\boldsymbol{l}^{*},\boldsymbol{\Lambda}^{*})$
should satisfy 
\begin{equation}
\frac{\partial L}{\partial\phi\left(t\right)}=0\label{eq:stationary_c_o_1-1}
\end{equation}
\begin{equation}
\frac{\partial L}{\partial l\left(t\right)}=0\label{eq:stationary_c_o_2-1}
\end{equation}
\begin{equation}
-\phi\left(t\right)v_{1}\left(t\right)\ge0,\,v_{1}\left(t\right)\ge0\label{eq:complementary_slackness_c_o_1-1}
\end{equation}
\begin{equation}
\left(\phi\left(t\right)-\bar{c}\left(t\right)l\left(t\right)\right)v_{2}\left(t\right)\ge0,\,v_{2}\left(t\right)\ge0\label{eq:complementary_slackness_c_o_2-1}
\end{equation}
\begin{equation}
-l\left(t\right)v_{3}\left(t\right)\ge0,\,v_{3}\left(t\right)\ge0\label{eq:complementary_slackness_c_o_3-1}
\end{equation}
\begin{equation}
\left(l\left(t\right)-1\right)v_{4}\left(t\right)\ge0,\,v_{4}\left(t\right)\ge0\label{eq:complementary_slackness_c_o_4-1}
\end{equation}
\begin{equation}
\left(S-\sum_{t\in\mathcal{T}}\phi\left(t\right)\right)\mu\ge0,\,\mu\ge0.\label{eq:complementary_slackness_c_o_7-1}
\end{equation}

Re-converting $\phi\left(t\right)$ by $c\left(t\right)$ and $l\left(t\right)$
according to $c\left(t\right)=\phi\left(t\right)/l\left(t\right)$,
based on (\ref{eq:stationary_c_o_1-1})-(\ref{eq:complementary_slackness_c_o_7-1}),
we have that the optimal capacity and frequency allocation $\left(\boldsymbol{c}^{*},\boldsymbol{l}^{*}\right)$
should satisfy $\forall t$, (\ref{eq:complementary_slackness_c_o_3-1},
\ref{eq:complementary_slackness_c_o_4-1}) and 
\begin{equation}
\ln2\cdot2^{c\left(t\right)+\epsilon\left(t\right)}/g\left(t\right)-\mu-v_{1}\left(t\right)+v_{2}\left(t\right)=0\label{eq:stationary_c_1-1}
\end{equation}
\begin{equation}
\vartheta\left(c(t)\right)-v_{2}\left(t\right)\bar{c}\left(t\right)-v_{3}\left(t\right)+v_{4}\left(t\right)=0\label{eq:stationary_c_2-1}
\end{equation}
\begin{equation}
-c\left(t\right)l\left(t\right)v_{1}\left(t\right)\ge0,\,v_{1}\left(t\right)\ge0\label{eq:complementary_slackness_c_1-1}
\end{equation}
\begin{equation}
\left(c\left(t\right)l\left(t\right)-\bar{c}\left(t\right)l\left(t\right)\right)v_{2}\left(t\right)\ge0,\,v_{2}\left(t\right)\ge0\label{eq:complementary_slackness_c_2-1}
\end{equation}
\begin{equation}
\left(S-\sum_{t\in\mathcal{T}}c\left(t\right)l\left(t\right)\right)\mu\ge0,\,\mu\ge0.\label{eq:complementary_slackness_c_7-1}
\end{equation}

\subsection{Optimal allocation strategy}

Based on conditions (\ref{eq:stationary_c_1-1}, \ref{eq:complementary_slackness_c_1-1},
\ref{eq:complementary_slackness_c_2-1}), we can get the relationship
between the optimal capacity $c^{*}(t)$ and the optimal Lagrangian
multiplier $\mu^{*}$ during the non-zero allocated time $l^{*}(t)$.
Note that when $l^{*}(t)$ is $0$, the capacity $c^{*}(t)$ is $0$,
according to the definition of $p(t)$ in (\ref{eq:p_over_phi_l}).

If $c^{*}(t)=0$, then $v_{1}\left(t\right)\ge0$, $v_{2}\left(t\right)=0$
according to (\ref{eq:complementary_slackness_c_1-1}) and (\ref{eq:complementary_slackness_c_2-1}).
Then, based on condition (\ref{eq:stationary_c_1-1}), we have $\mu^{*}\le\ln2\cdot2^{c^{*}(t)+\epsilon(t)}/g(t)=\ln2\cdot2^{\epsilon\left(t\right)}/g\left(t\right)$.
If $c^{*}(t)\in(0,\bar{c}(t))$, then $v_{1}\left(t\right)=0$, $v_{2}\left(t\right)=0$
according to (\ref{eq:complementary_slackness_c_1-1}) and (\ref{eq:complementary_slackness_c_2-1}).
Then, based on condition (\ref{eq:stationary_c_1-1}), we have $\mu^{*}=\ln2\cdot2^{c^{*}\left(t\right)+\epsilon\left(t\right)}/g\left(t\right)$.
If $c^{*}(t)=\bar{c}(t)$, then $v_{1}\left(t\right)=0$, $v_{2}\left(t\right)\ge0$
according to (\ref{eq:complementary_slackness_c_1-1}) and (\ref{eq:complementary_slackness_c_2-1}).
Then, based on condition (\ref{eq:stationary_c_1-1}), we have $\mu^{*}\ge\ln2\cdot2^{c^{*}\left(t\right)+\epsilon\left(t\right)}/g\left(t\right)=\ln2\cdot2^{\bar{c}\left(t\right)+\epsilon\left(t\right)}/g\left(t\right)$.
In summary, the optimal allocated capacity can be expressed as a function
of optimal Lagrangian multiplier $\mu^{*}$, as
\begin{equation}
c^{*}\left(t\right)=\begin{cases}
0 & \mu^{*}\le\underline{\mu}\left(t\right)\\
\log\left(\mu^{*}g\left(t\right)/\ln2\right)-\epsilon\left(t\right) & \underline{\mu}\left(t\right)<\mu^{*}<\overline{\mu}\left(t\right)\\
\bar{c}\left(t\right) & \mu^{*}\ge\overline{\mu}\left(t\right)
\end{cases}\label{eq:opt_c-1}
\end{equation}
where $\underline{\mu}(t)=\ln2\cdot2^{\epsilon(t)}/g(t)$ and $\overline{\mu}(t)=\ln2\cdot2^{\bar{c}(t)+\epsilon(t)}/g(t)$.

Based on conditions (\ref{eq:complementary_slackness_c_o_3-1}, \ref{eq:complementary_slackness_c_o_4-1},
\ref{eq:stationary_c_2-1}, \ref{eq:complementary_slackness_c_2-1}),
we can get the relationship between the optimal time $l^{*}(t)$ and
the optimal capacity $c^{*}(t)$.

When $c^{*}\left(t\right)=0$, there must be $l^{*}\left(t\right)=0$.
Otherwise, $v_{2}\left(t\right)=0$ according to (\ref{eq:complementary_slackness_c_2-1})
and $v_{3}\left(t\right)=0$ according to (\ref{eq:complementary_slackness_c_o_3-1}),
then condition (\ref{eq:stationary_c_2-1}) becomes 
\[
(2^{\epsilon\left(t\right)}-1)/g\left(t\right)+\lambda+v_{4}\left(t\right)=0
\]
thus $(2^{\epsilon(t)}-1)/g(t)+\lambda\le0$, then $\lambda<0$, which
contradicts the setting of $\lambda$ greater than 0.

When $c^{*}\left(t\right)>0$, there are two cases based on whether
the maximum generalized efficiency is achievable or not, where the
generalized efficiency function is defined as $\eta(t)\triangleq c(t)/\varphi(c(t))$,
where $\varphi(c(t))=(2^{c(t)+\epsilon\left(t\right)}-1)/g(t)+\lambda$
represents the overall cost consumption. Since $\varphi(c(t))$ is
convex over $c(t)$, then $\eta(t)$ is a quasi-concave function \ac{wrt}
$c(t)$ \cite{boyd2004convex}, and the maximal generalized efficiency
occurs where its first derivative $\frac{\partial\eta(t)}{\partial c(t)}$
equals to $0$,
\[
\frac{\left(2^{c\left(t\right)+\epsilon\left(t\right)}-1\right)/g\left(t\right)+\lambda-c\left(t\right)\left(\ln2\cdot2^{c\left(t\right)+\epsilon\left(t\right)}/g\left(t\right)\right)}{\left(\left(2^{c\left(t\right)+\epsilon\left(t\right)}-1\right)/g\left(t\right)+\lambda\right)^{2}}=0
\]
that is, $\vartheta(c(t))/((2^{c\left(t\right)+\epsilon\left(t\right)}-1)/g(t)+\lambda)^{2}=0$
according to the definition of $\vartheta(c(t))$ in (\ref{eq:def_theta_c}).
Since $((2^{c(t)+\epsilon(t)}-1)/g(t)+\lambda)^{2}>0$ due to $c(t)\ge0$,
$\epsilon\left(t\right)\ge0$, and $\lambda>0$, the solution to the
equation $\vartheta(c(t))=0$ with $c(t)\ge0$ is the capacity with
maximal generalized efficiency, denoted as $\hat{c}(t)$. Note that
$\hat{c}(t)$ always exists and $\hat{c}(t)>0$ according to Lemma
\ref{lem:unique_root}. In other words,
\begin{equation}
\vartheta\left(c\left(t\right)\right)\begin{cases}
>0 & c\left(t\right)\in[0,\hat{c}(t))\\
=0 & c\left(t\right)=\hat{c}(t)\\
<0 & c\left(t\right)>\hat{c}(t).
\end{cases}\label{eq:vartheta_rl-1}
\end{equation}

For the case (i) when the maximum generalized efficiency is achievable,
that is, $t\in\mathcal{T}_{1}\triangleq\{t\in\mathcal{T}:\hat{c}(t)\le\bar{c}(t)\}$.
If $c(t)\in[0,\hat{c}(t))$, we have $\vartheta(c(t))>0$ according
to (\ref{eq:vartheta_rl-1}) and $v_{2}(t)=0$ according to (\ref{eq:complementary_slackness_c_2-1}).
Then based on (\ref{eq:stationary_c_2-1}), we have
\[
-v_{3}\left(t\right)+v_{4}\left(t\right)<0
\]
which indicates $v_{3}\left(t\right)\neq0$, thus $l^{*}\left(t\right)=0$
according to (\ref{eq:complementary_slackness_c_o_3-1}).

If $c\left(t\right)=\hat{c}(t)$, we have $\vartheta(c(t))=0$ according
to (\ref{eq:vartheta_rl-1}) and $v_{2}(t)=0$ according to (\ref{eq:complementary_slackness_c_2-1}).
Then based on (\ref{eq:stationary_c_2-1}), we have
\[
-v_{3}\left(t\right)+v_{4}\left(t\right)=0
\]
which indicates $v_{3}(t)=v_{4}(t)\ge0$, thus $l^{*}\left(t\right)\in[0,1]$
according to (\ref{eq:complementary_slackness_c_o_3-1}) and (\ref{eq:complementary_slackness_c_o_4-1}).

If $c\left(t\right)=(\hat{c}(t),\bar{c}(t)]$, we have $\vartheta(c(t))<0$
according to (\ref{eq:vartheta_rl-1}). Then based on (\ref{eq:stationary_c_2-1}),
we have
\[
-v_{2}\left(t\right)\bar{c}\left(t\right)-v_{3}\left(t\right)+v_{4}\left(t\right)>0
\]
which indicates $v_{4}(t)\neq0$, thus $l^{*}(t)=1$ according to
(\ref{eq:complementary_slackness_c_o_4-1}).

In summary, the relationship between between the optimal time $l^{*}(t)$
and the optimal capacity $c^{*}(t)$ is
\begin{equation}
l^{*}\left(t\right)\begin{cases}
=0 & c^{*}\left(t\right)\in[0,\hat{c}\left(t\right))\\
\in[0,1] & c^{*}\left(t\right)=\hat{c}\left(t\right)\\
=1 & c^{*}\left(t\right)\in(\hat{c}\left(t\right),\bar{c}\left(t\right)]
\end{cases}\label{eq:opt_l_case1-1}
\end{equation}
for $t\in\mathcal{T}_{1}\triangleq\left\{ t:\hat{c}\left(t\right)\le\bar{c}\left(t\right)\right\} $.

Similarly, for $t\in\mathcal{T}_{2}\triangleq\left\{ t:\hat{c}\left(t\right)>\bar{c}\left(t\right)\right\} $,
we have
\begin{equation}
l^{*}\left(t\right)\begin{cases}
=0 & c^{*}\left(t\right)\in\left[0,\bar{c}\left(t\right)\right)\\
\in[0,1] & c^{*}\left(t\right)=\bar{c}\left(t\right).
\end{cases}\label{eq:opt_l_case2-1}
\end{equation}

Finally, according to condition (\ref{eq:complementary_slackness_c_7-1}),
the optimal solution $(c^{*}(t),l^{*}(t))$ is the solution to the
following problem
\begin{equation}
\sum_{t\in\mathcal{T}}c\left(t\right)l\left(t\right)=S\label{eq:thp_c_eq-1}
\end{equation}
otherwise, $\mu=0$, then $c^{*}\left(t\right)=0$ according to (\ref{eq:opt_c-1})
and the throughput constraint cannot satisfy.

\subsection{Dual problem}

Denote $\tilde{\bm{l}}\triangleq[\tilde{l}(t)]_{t\in\mathcal{T}}$
as the the frequency allocation when $\mu=\hat{\mu}(t)$ for $t\in\mathcal{T}_{1}$
and when $\mu=\tilde{\mu}(t)$ for $t\in\mathcal{T}_{2}$. Combine
the optimal capacity allocation strategy (\ref{eq:complementary_slackness_c_o_4-1})
and the optimal frequency allocation strategy, (\ref{eq:opt_l_case1-1})
and (\ref{eq:opt_l_case2-1}), the optimal $\phi(t)$ and $l(t)$
can be expressed as

(i) for $t\in\mathcal{T}_{1}$
\begin{equation}
\phi\left(t\right)=\begin{cases}
0 & \mu<\hat{\mu}\left(t\right)\\
\hat{c}\left(t\right)\tilde{l}\left(t\right) & \mu=\hat{\mu}\left(t\right)\\
\log\left(\mu g\left(t\right)/\ln2\right)-\epsilon\left(t\right) & \hat{\mu}\left(t\right)<\mu<\bar{\mu}\left(t\right)\\
\bar{c}\left(t\right) & \mu\ge\bar{\mu}\left(t\right)
\end{cases}\label{eq:ct_1_A-1}
\end{equation}
\begin{equation}
l\left(t\right)=\mathbb{I}\left\{ \mu>\hat{\mu}\left(t\right)\right\} +\tilde{l}\left(t\right)\mathbb{I}\left\{ \mu=\hat{\mu}\left(t\right)\right\} \label{eq:lt_1_A-1}
\end{equation}

(ii) for $t\in\mathcal{T}_{2}$
\begin{equation}
\phi\left(t\right)=\begin{cases}
0 & \mu<\tilde{\mu}\left(t\right)\\
\bar{c}\left(t\right)\tilde{l}\left(t\right) & \mu=\tilde{\mu}\left(t\right)\\
\bar{c}\left(t\right) & \mu>\tilde{\mu}\left(t\right)
\end{cases}\label{eq:ct_2_A-1}
\end{equation}
\begin{equation}
l\left(t\right)=\mathbb{I}\left\{ \mu>\tilde{\mu}\left(t\right)\right\} +\tilde{l}\left(t\right)\mathbb{I}\left\{ \mu=\tilde{\mu}\left(t\right)\right\} \label{eq:lt_2_A-1}
\end{equation}
where (\ref{eq:ct_1_A-1}) holds because $\underline{\mu}(t)<\hat{\mu}(t)\triangleq\ln2\cdot2^{\hat{c}(t)+\epsilon(t)}/g(t)$
due to $\hat{c}(t)>0$, and (\ref{eq:ct_2_A-1}) holds because $\tilde{\mu}(t)=(\bar{p}(t)+\lambda)/(\log(1+\bar{p}(t)g(t))-\epsilon(t))$
is the marginal rate of throughput cost when the slot is active at
time $t$ for $t\in\mathcal{T}_{2}$. Note that for $t\in\mathcal{T}_{1}$,
the marginal rate of throughput cost, denoted as $\Delta\tilde{F}(t)/\Delta\phi(t)$,
is always $\mu$, according to Lemma \ref{lem:marginal_rate-1}.

Finally, combining the condition (\ref{eq:thp_c_eq-1}), the parameter
$(\mu,\tilde{\bm{l}})$ are chosen from 
\[
\tilde{\Upsilon}\left(\mu,\tilde{\bm{l}}\right)\triangleq\sum_{t\in\mathcal{T}}\phi(t)=S.
\]

\begin{lem}
\label{lem:marginal_rate-1}For $t\in\mathcal{T}_{1}$, the marginal
rate of throughput cost is always $\mu$.
\end{lem}
\begin{proof}
There are two kinds of throughput increasing for $t\in\mathcal{T}_{1}$,
i) increasing the allocated time $\Delta t$ when $\mu=\hat{\mu}\left(t\right)$;
ii) increasing the allocated power $\Delta p$ when $\hat{\mu}\left(t\right)>\mu>\bar{\mu}\left(t\right)$.
For case i), the marginal rate of throughput cost is 
\begin{equation}
\frac{\Delta\tilde{F}(t)}{\Delta\phi(t)}=\frac{\left(\hat{p}\left(t\right)+\lambda\right)\cdot\Delta t}{\hat{c}\left(t\right)\cdot\Delta t}=\frac{(2^{\hat{c}(t)+\epsilon\left(t\right)}-1)/g(t)+\lambda}{\hat{c}\left(t\right)}.\label{eq:c_per_thp_v1-1}
\end{equation}
Since $\vartheta(\hat{c}(t))=0$ according to (\ref{eq:vartheta_rl-1}),
we have $(2^{\hat{c}(t)+\epsilon\left(t\right)}-1)/g(t)+\lambda=\hat{c}(t)\cdot(\ln2\cdot2^{\hat{c}(t)+\epsilon(t)}/g(t))$.
Thus (\ref{eq:c_per_thp_v1-1}) becomes
\[
\frac{\Delta\tilde{F}(t)}{\Delta\phi(t)}=\frac{\ln2\cdot2^{\hat{c}(t)+\epsilon(t)}}{g(t)}\stackrel{(a)}{=}\hat{\mu}\left(t\right)=\mu
\]
where (a) holds because $\hat{c}(t)=\log(\hat{\mu}(t)g(t)/\ln2)-\epsilon(t)$
according to (\ref{eq:opt_c-1}).

For case ii), the marginal rate of throughput cost is 
\begin{equation}
\frac{\Delta\tilde{F}(t)}{\Delta\phi(t)}=\frac{\nabla_{\mu}p\left(t\right)\cdot\Delta\mu}{\nabla_{\mu}c\left(t\right)\cdot\Delta\mu}\label{eq:c_per_thp_v2-1}
\end{equation}
where $c(t)=\log\left(\mu\cdot g(t)/\ln2\right)-\epsilon(t)$ and
$p(t)=(2^{c(t)+\epsilon\left(t\right)}-1)/g(t)=(\mu g(t)/\ln2-1)/g(t)$
according to (\ref{eq:opt_c-1}). Thus (\ref{eq:c_per_thp_v2-1})
becomes 
\[
\frac{\Delta\tilde{F}(t)}{\Delta\phi(t)}=\frac{1\cdot\Delta\mu}{\frac{1}{\mu}\cdot\Delta\mu}=\mu.
\]
\end{proof}

\section{Proof of Proposition \ref{prop:monotonicity}\label{sec:proof_prop_monotonicity}}

Function $\tilde{\Upsilon}(\mu,\tilde{\bm{l}})$ is a summation of
non-negative function $\phi(t)$ over $t$ according to its definition
in (\ref{eq:dual-obj}). Thus, function $\tilde{\Upsilon}(\mu,\tilde{\bm{l}})$
has monotonicity property if function $\phi(t;\mu,l)$ over $(\mu,l)$
for all $t\in\mathcal{T}$ has monotonicity property. Accordingly,
we will prove that for all $t\in\mathcal{T}$, any $\mu_{1}>\mu_{2}\ge0$,
and $1\ge\tilde{l}_{1},\tilde{l}_{2}\ge0$, it holds that $\phi(t;\mu_{1},\tilde{l}_{1})\ge\phi(t;\mu_{2},\tilde{l}_{2})$;
then, prove that for all $t\in\mathcal{T}$, any $1\ge\tilde{l}_{1}>\tilde{l}_{2}\ge0$
and $\mu\ge0$, it holds that $\phi(t;\mu,\tilde{l}_{1})\ge\phi(t;\mu,\tilde{l}_{2})$.

When $\mu_{1}>\mu_{2}\ge\bar{\mu}(t)$, $\phi(t;\mu_{1},\tilde{l}_{1})=\phi(t;\mu_{2},\tilde{l}_{2}),\,\forall\tilde{l}_{1},\tilde{l}_{2}\in[0,1]$.
When $\bar{\mu}(t)>\mu_{1}>\mu_{2}>\hat{\mu}(t)$, the throughput
$\log(\mu g(t)/\ln2)-\epsilon(t)$ is positive and monotonically increasing
over $\mu$ due to its first derivative $(\ln2\cdot\mu)^{-1}>0$,
thus $\phi(t;\mu_{1},\tilde{l}_{1})>\phi(t;\mu_{2},\tilde{l}_{2}),\,\forall\tilde{l}_{1},\tilde{l}_{2}\in[0,1]$.
When $\hat{\mu}(t)\ge\mu_{1}>\mu_{2}\ge0$, $\lim_{\mu\to\hat{\mu}(t)^{-}}\phi(t)=0\le\phi(t;\hat{\mu}(t),\tilde{l})\le\log(\hat{\mu}(t)g(t)/\ln2)-\epsilon(t)=\lim_{\mu\to\hat{\mu}(t)^{+}}\phi(t)$
due to $\tilde{l}\in[0,1]$, thus $\phi(t;\mu_{1},\tilde{l}_{1})\ge\phi(t;\mu_{2},\tilde{l}_{2}),\,\forall\tilde{l}_{1},\tilde{l}_{2}\in[0,1]$.\footnote{Here, $a^{-}$ represents the left-sided limit and $a^{+}$ represents
the right-sided limit.} In summary, for all $t\in\mathcal{T}$, any $\mu_{1}>\mu_{2}\ge0$,
and $1\ge\tilde{l}_{1},\tilde{l}_{2}\ge0$, it holds that $\phi(t;\mu_{1},\tilde{l}_{1})\ge\phi(t;\mu_{2},\tilde{l}_{2})$.

If $\mu\neq\hat{\mu}(t)$, $\tilde{l}$ has no effect on $\phi$,
then we have $\phi(t;\mu,\tilde{l}_{1})=\phi(t;\mu,\tilde{l}_{2})$
for any $1\ge\tilde{l}_{1}>\tilde{l}_{2}\ge0$. If $\mu\neq\hat{\mu}(t)$,
$\phi$ is linear function over $\tilde{l}$, thus we have $\phi(t;\mu,\tilde{l}_{1})>\phi(t;\mu,\tilde{l}_{2})$
for any $1\ge\tilde{l}_{1}>\tilde{l}_{2}\ge0$. In summary, for all
$t\in\mathcal{T}$, any $1\ge\tilde{l}_{1}>\tilde{l}_{2}\ge0$ and
$\mu\ge0$, it holds that $\phi(t;\mu,\tilde{l}_{1})\ge\phi(t;\mu,\tilde{l}_{2})$.

\section{Large Gap Case\label{sec:proof-of-gap-23}}

Consider a one-receiver zero-neighbor scenario with constant and deterministic
channel, that is, $N=1$, $M=0$, $g(t)=g_{\text{c}}$, and $\kappa(t)=\infty$
for all $t\in\mathcal{T}$. One of the optimal solutions to $\mathscr{P}2$
is $l^{\prime}(t)=l_{\text{c}}$ and $p^{\prime}(t)=p_{\text{c}}$
for all $t\in\mathcal{T}$, where $Tl_{c}\log(1+p_{c}g_{c})=S$. Suppose
there is a throughput requirement $S$ such that $\sum_{t\in\mathcal{T}}l^{\prime}(t)\le1$,
that is, $l_{\text{c}}T\le1$, one can construct another transmission
strategy $(\boldsymbol{P}^{\prime\prime},\boldsymbol{L}^{\prime\prime})$,
where $l^{\prime\prime}(1)=l_{\text{c}}T$, $p^{\prime\prime}(1)=p_{\text{c}}$,
and $l^{\prime\prime}(t)=0$, $p^{\prime\prime}(t)=0$ for $t\in\mathcal{T}\backslash\{1\}$.

The following can be observed for strategy $(\boldsymbol{P}^{\prime\prime},\boldsymbol{L}^{\prime\prime})$:
(i) the strategy $(\boldsymbol{P}^{\prime\prime},\boldsymbol{L}^{\prime\prime})$
is also an optimal solution to problem $\mathscr{P}2$ because the
relaxed cost (\ref{eq:erg_relax}) and throughput (\ref{eq:thp_n_def})
under $(\boldsymbol{P}^{\prime\prime},\boldsymbol{L}^{\prime\prime})$
are the same as under $(\boldsymbol{P}^{\prime},\boldsymbol{L}^{\prime})$.
(ii) The difference of the actual cost (\ref{eq:erg_def}) of two
strategies is $\lambda(T-1)$ because 
\begin{align*}
 & F\left(\bm{P}^{\prime},\bm{L}^{\prime}\right)-F\left(\bm{P}^{\prime\prime},\bm{L}^{\prime\prime}\right)\\
 & =\sum_{t\in\mathcal{T}}p^{\prime}\left(t\right)l^{\prime}\left(t\right)+\lambda\sum_{t\in\mathcal{T}}\mathbb{I}\left\{ l^{\prime}\left(t\right)>0\right\} \\
 & \quad-\sum_{t\in\mathcal{T}}p^{\prime\prime}\left(t\right)l^{\prime\prime}\left(t\right)+\lambda\sum_{t\in\mathcal{T}}\mathbb{I}\left\{ l^{\prime\prime}\left(t\right)>0\right\} \\
 & =\lambda T-\lambda.
\end{align*}

In addition, the optimal value of $\mathscr{P}1$, denoted as $F^{*}$,
is less than or equal to the cost in (\ref{eq:erg_def}) by the solution
to $\mathscr{P}2$, {\em i.e.}, $F^{*}\le F(\boldsymbol{P}^{*},\boldsymbol{L}^{*})$,
where $(\boldsymbol{P}^{*},\boldsymbol{L}^{*})$ is an optimal solution
to $\mathscr{P}2$, because $\mathscr{P}1$ and $\mathscr{P}2$ are
two optimization problems with the same feasible set. As a result,
the additional cost in (\ref{eq:erg_def}) by the solution of $\mathscr{P}2$
can be given by
\begin{align*}
F\left(\boldsymbol{P}^{*},\boldsymbol{L}^{*}\right)-F^{*} & \ge F\left(\boldsymbol{P}^{*},\boldsymbol{L}^{*}\right)-F\left(\bm{P}^{\prime\prime},\bm{L}^{\prime\prime}\right)\\
 & \stackrel{(a)}{=}F\left(\bm{P}^{\prime},\bm{L}^{\prime}\right)-F\left(\bm{P}^{\prime\prime},\bm{L}^{\prime\prime}\right)\\
 & =\lambda\left(T-1\right).
\end{align*}
where (a) holds because $(\boldsymbol{P}^{\prime},\boldsymbol{L}^{\prime})$
is one of the optimal solution of $\mathscr{P}2$.

\section{Proof of Proposition \ref{prop:opt_preserve_c}\label{sec:proof-lem-opt_preserve_c}}

Denote the the generalized efficiency on the optimal solution as 
\begin{equation}
\eta_{n}\left(t\right)=c_{n}^{*}\left(t\right)/\left(p_{n}^{*}\left(t\right)+\lambda\right)\label{eq:erg-efficiency-def-1}
\end{equation}
where $c_{n}^{*}(t)=\text{\ensuremath{\log}}(1+p_{n}^{*}(t)g_{n}(t))-\epsilon_{n}(t)$.
\begin{lem}
\label{lem:constant_effc-1}The generalized efficiencies of all partially-used
slots by any receiver $n$ are the same, that is, $\eta_{n}(t_{1})=\eta_{n}(t_{2}),\,\forall t_{1},t_{2}\in\bar{\mathcal{S}}_{n}(\boldsymbol{L}^{*})$.
\end{lem}
\begin{proof}
Suppose that there is an optimal solution to $\mathscr{P}2$, $(\bm{P}^{*},\bm{L}^{*})$,
where there are two partially-used slots $t_{1},t_{2}\in\bar{\mathcal{S}}_{n}$
with different generalized efficiencies, \ac{wlog}, $\eta_{n}(t_{1})>\eta_{n}(t_{2})$.
One can construct another frequency allocation strategy $\bm{L}^{\prime}$,
where $l_{n}^{\prime}(t_{1})=l_{n}^{*}(t_{1})+l_{\text{a}}$, $l_{n}^{\prime}(t_{2})=l_{n}^{*}(t_{2})-l_{\text{s}}$,
and $l_{n}^{\prime}(t)=l_{n}^{*}(t)$ for all $t\in\mathcal{T}\backslash\{t_{1},t_{2}\}$,
such that $\Upsilon_{n}(\bm{P}^{*},\bm{L}^{*})=\Upsilon_{n}(\bm{P}^{*},\bm{L}^{\prime})$,
that is
\begin{align*}
 & \eta_{n}\left(t_{1}\right)\left(p_{n}^{*}\left(t_{1}\right)+\lambda\right)l_{n}^{*}\left(t_{1}\right)+\eta_{n}\left(t_{2}\right)\left(p_{n}^{*}\left(t_{2}\right)+\lambda\right)l_{n}^{*}\left(t_{2}\right)\\
 & =\eta_{n}\left(t_{1}\right)\left(p_{n}^{*}\left(t_{1}\right)+\lambda\right)l_{n}^{\prime}\left(t_{1}\right)+\eta_{n}\left(t_{2}\right)\left(p_{n}^{*}\left(t_{2}\right)+\lambda\right)l_{n}^{\prime}\left(t_{2}\right)
\end{align*}
based on the definition of generalized efficiency in (\ref{eq:erg-efficiency-def-1}).
Thus we have $\eta_{n}\left(t_{1}\right)\left(p_{n}^{*}\left(t_{1}\right)+\lambda\right)l_{\text{a}}=\eta_{n}\left(t_{2}\right)\left(p_{n}^{*}\left(t_{2}\right)+\lambda\right)l_{\text{s}}$.

However, the cost consumed by strategy $(\bm{P}^{*},\bm{L}^{\prime})$
is less than by $(\bm{P}^{*},\bm{L}^{*})$. This is due to the less
cost for transmitting to receiver $n$
\begin{align}
 & F_{n}(\bm{P}^{*},\bm{L}^{\prime})\nonumber \\
 & =\left(p_{n}^{*}\left(t_{1}\right)+\lambda\right)l_{n}^{\prime}\left(t_{1}\right)+\left(p_{n}^{*}\left(t_{2}\right)+\lambda\right)l_{n}^{\prime}\left(t_{2}\right)+E_{\text{r}}\nonumber \\
 & =\left(p_{n}^{*}\left(t_{1}\right)+\lambda\right)l_{n}^{*}\left(t_{1}\right)+\left(p_{n}^{*}\left(t_{2}\right)+\lambda\right)l_{n}^{*}\left(t_{2}\right)+E_{\text{r}}\nonumber \\
 & \quad+\left(p_{n}^{*}\left(t_{1}\right)+\lambda\right)l_{\text{a}}-\left(p_{n}^{*}\left(t_{2}\right)+\lambda\right)l_{\text{s}}\nonumber \\
 & =E_{n}(\bm{P}^{*},\bm{L}^{*})+\left(p_{n}^{*}\left(t_{2}\right)+\lambda\right)l_{\text{s}}\left(\eta_{n}\left(t_{2}\right)/\eta_{n}\left(t_{1}\right)-1\right)\nonumber \\
 & \stackrel{(a)}{<}F_{n}(\bm{P}^{*},\bm{L}^{*})\label{eq:erg_dif_1-1}
\end{align}
where $E_{\text{r}}\triangleq\sum_{t\in\mathcal{T}\backslash\{t_{1},t_{2}\}}(p_{n}^{*}(t)+\lambda)l_{n}^{*}(t)$,
and (a) holds because $\eta_{n}(t_{1})>\eta_{n}(t_{2})$. Thus, $(\bm{P}^{*},\bm{L}^{*})$
is a strictly sub-optimal, which contradicts the hypothesis. As a
result, for all $n\in\mathcal{N}$, the generalized efficiencies $\eta_{n}(t)$
of all partially-used slots $t\in\bar{\mathcal{S}}_{n}$ are the same.
\end{proof}
Then, we prove that the strategy $(\boldsymbol{P},\boldsymbol{L})$
is also optimal by showing the cost and throughput are the same for
both strategies $(\boldsymbol{P}^{*},\boldsymbol{L}^{*})$ and $(\boldsymbol{P},\boldsymbol{L})$.

The throughput on any receiver $n$ under $(\bm{P},\bm{L})$ is the
same as under $(\bm{P}^{*},\bm{L}^{*})$ because
\begin{align*}
\Upsilon_{n}\left(\boldsymbol{P}^{*},\boldsymbol{L}^{*}\right) & =\sum_{t\in\bar{\mathcal{S}}_{n}^{*}}\phi_{n}^{*}\left(t\right)+\sum_{t\in\mathcal{T}\backslash\bar{\mathcal{S}}_{n}^{*}}\phi_{n}^{*}\left(t\right)\\
 & \stackrel{(a)}{=}\sum_{t\in\bar{\mathcal{S}}_{n}^{*}}\phi_{n}^{*}\left(t\right)+\sum_{t\in\mathcal{T}\backslash\bar{\mathcal{S}}_{n}^{*}}\phi_{n}\left(t\right)\\
 & \stackrel{(b)}{=}\sum_{t\in\bar{\mathcal{S}}_{n}^{*}}\phi_{n}\left(t\right)+\sum_{t\in\mathcal{T}\backslash\bar{\mathcal{S}}_{n}^{*}}\phi_{n}\left(t\right)\\
 & =\Upsilon_{n}\left(\boldsymbol{P},\boldsymbol{L}\right)
\end{align*}
where $\bar{\mathcal{S}}_{n}^{*}\triangleq\bar{\mathcal{S}}_{n}(\boldsymbol{L}^{*})$,
(a) holds because $p_{n}(t)=p_{n}^{*}(t)$ and $l_{n}(t)=l_{n}^{*}(t)$
for all $t\in\mathcal{T}\backslash\bar{\mathcal{S}}_{n}^{*}$ according
to condition (i) and (ii), and (b) holds because 
\begin{align*}
\sum_{t\in\bar{\mathcal{S}}_{n}^{*}}\phi_{n}^{*}\left(t\right) & =\sum_{t\in\bar{\mathcal{S}}_{n}^{*}}l_{n}^{*}\left(t\right)c_{n}^{*}\left(t\right)\\
 & \stackrel{(c)}{=}\sum_{t\in\bar{\mathcal{S}}_{n}^{*}}l_{n}^{*}\left(t\right)\eta_{n}\left(t\right)\left(p_{n}^{*}\left(t\right)+\lambda\right)\\
 & \stackrel{(d)}{=}\bar{\eta}_{n}\sum_{t\in\bar{\mathcal{S}}_{n}^{*}}l_{n}^{*}\left(t\right)\left(p_{n}^{*}\left(t\right)+\lambda\right)\\
 & \stackrel{(e)}{=}\bar{\eta}_{n}\sum_{t\in\bar{\mathcal{S}}_{n}^{*}}l_{n}\left(t\right)\left(p_{n}\left(t\right)+\lambda\right)\\
 & \stackrel{(f)}{=}\sum_{t\in\bar{\mathcal{S}}_{n}^{*}}\eta_{n}\left(t\right)l_{n}\left(t\right)\left(p_{n}\left(t\right)+\lambda\right)\\
 & \stackrel{(g)}{=}\sum_{t\in\bar{\mathcal{S}}_{n}^{*}}l_{n}\left(t\right)c_{n}\left(t\right)=\sum_{t\in\bar{\mathcal{S}}_{n}^{*}}\phi_{n}\left(t\right)
\end{align*}
where (c) and (g) hold according to (\ref{eq:erg-efficiency-def-1}),
(d) and (f) holds according to Lemma \ref{lem:constant_effc-1}, and
(e) holds due to condition (iii).

The relaxed cost in (\ref{eq:erg_relax}) under $(\bm{P},\bm{L})$
is the same as under $(\bm{P}^{*},\bm{L}^{*})$ because 
\begin{align*}
 & \tilde{F}\left(\boldsymbol{P}^{*},\boldsymbol{L}^{*}\right)\\
 & =\sum_{n=1}^{N}\left(\sum_{t\in\bar{\mathcal{S}}_{n}^{*}}\left(p_{n}^{*}\left(t\right)+\lambda\right)l_{n}^{*}\left(t\right)+\sum_{t=\mathcal{T}\backslash\bar{\mathcal{S}}_{n}^{*}}\left(p_{n}^{*}\left(t\right)+\lambda\right)l_{n}^{*}\left(t\right)\right)\\
 & \stackrel{(a)}{=}\sum_{n=1}^{N}\left(\sum_{t\in\bar{\mathcal{S}}_{n}^{*}}\left(p_{n}^{*}\left(t\right)+\lambda\right)l_{n}^{*}\left(t\right)+\sum_{t=\mathcal{T}\backslash\bar{\mathcal{S}}_{n}^{*}}\left(p_{n}\left(t\right)+\lambda\right)l_{n}\left(t\right)\right)\\
 & \stackrel{(b)}{=}\sum_{n=1}^{N}\left(\sum_{t\in\bar{\mathcal{S}}_{n}^{*}}\left(p_{n}\left(t\right)+\lambda\right)l_{n}\left(t\right)+\sum_{t=\mathcal{T}\backslash\bar{\mathcal{S}}_{n}^{*}}\left(p_{n}\left(t\right)+\lambda\right)l_{n}\left(t\right)\right)\\
 & =\tilde{F}\left(\boldsymbol{P}^{\prime},\boldsymbol{L}^{\prime}\right)
\end{align*}
where (a) holds because $p_{n}(t)=p_{n}^{*}(t)$ and $l_{n}(t)=l_{n}^{*}(t)$
for all $t\in\mathcal{T}\backslash\bar{\mathcal{S}}_{n}^{*}$ according
to condition (i) and (ii), and (b) holds according to condition (iii).

As a result, the strategy $(\boldsymbol{P},\boldsymbol{L})$ is also
optimal if the three conditions in Proposition \ref{prop:opt_preserve_c}
hold.

\section{Proof of Proposition \ref{prop:perf_cap_non_r}\label{sec:proof-prop-perf-cap-non-r}}

The performance gap is bounded by 
\begin{align*}
 & F\left(\bm{P}^{*},\bm{L}^{*}\right)-F^{*}\\
 & \stackrel{(a)}{\le}F\left(\bm{P}^{*},\bm{L}^{*}\right)-\tilde{F}\left(\bm{P}^{*},\bm{L}^{*}\right)\\
 & =\lambda\sum_{t\in\bar{\mathcal{S}}\left(\bm{L}^{*}\right)}\left(\mathbb{I}\left\{ \sum_{n\in\mathcal{N}}l_{n}^{*}\left(t\right)>0\right\} -\sum_{n\in\mathcal{N}}l_{n}^{*}\left(t\right)\right)\\
 & \quad+\lambda\left(\sum_{t\in\mathcal{T}\backslash\bar{\mathcal{S}}\left(\bm{L}^{*}\right)}\mathbb{I}\left\{ \sum_{n\in\mathcal{N}}l_{n}^{*}\left(t\right)>0\right\} -\sum_{n\in\mathcal{N}}l_{n}^{*}\left(t\right)\right)\\
 & \stackrel{(b)}{=}\lambda\sum_{t\in\bar{\mathcal{S}}\left(\bm{L}^{*}\right)}\left(1-\sum_{n\in\mathcal{N}}l_{n}^{*}\left(t\right)\right)\\
 & \le\lambda\sum_{t\in\bar{\mathcal{S}}\left(\bm{L}^{*}\right)}\left(1\right)=\left|\bar{\mathcal{S}}\left(\bm{L}^{*}\right)\right|\lambda
\end{align*}
where (a) holds because the optimal value of problem $\mathscr{P}2$
is a lower bound on the optimal value of $\mathscr{P}1$, and (b)
holds because when $t\in\bar{\mathcal{S}}(\bm{L}^{*})$, then $\sum_{n\in\mathcal{N}}l_{n}^{*}(t)\in(0,1)$,
thus $\mathbb{I}\{\sum_{n\in\mathcal{N}}l_{n}^{*}(t)>0\}=1$, while
when $t\in\mathcal{T}\backslash\bar{\mathcal{S}}\left(\bm{L}^{*}\right)$,
then $\mathbb{I}\{\sum_{n\in\mathcal{N}}l_{n}^{*}(t)>0\}=\sum_{n\in\mathcal{N}}l_{n}^{*}(t)$
no matter of $\sum_{n\in\mathcal{N}}l_{n}^{*}(t)=0$ or $1$.

\section{Proof of Proposition \ref{prop:alg-opt}\label{sec:proof_gap_upper_bound}}

Denote the $n$th round of solution of Algorithm \ref{alg:Rounding-Algorithm}
as $\boldsymbol{L}^{(n)}$. We will prove by induction that the transmission
strategy $(\boldsymbol{P}^{*},\boldsymbol{L}^{(n)}),\forall n\in\mathcal{N}$
is optimal to $\mathscr{P}2$. Firstly, $(\boldsymbol{P}^{*},\boldsymbol{L}^{(0)})$
is a solution to $\mathscr{P}2$ because $(\boldsymbol{P}^{*},\boldsymbol{L}^{*})$
is a solution to $\mathscr{P}2$ and $\boldsymbol{L}^{(0)}=\boldsymbol{L}^{*}$.
Secondly, suppose $(\boldsymbol{P}^{*},\boldsymbol{L}^{(n-1)})$ is
an optimum to $\mathscr{P}2$ for all $n\in\{1,\cdots,N-1\}$, then
$(\boldsymbol{P}^{*},\boldsymbol{L}^{(n)})$ is also an optimum to
$\mathscr{P}2$ according to Proposition \ref{prop:opt_preserve_c},
because the constraints in problem $\mathscr{P}5$ ensure the conditions
in Proposition \ref{prop:opt_preserve_c} hold. As a result, the strategy
$(\boldsymbol{P}^{*},\hat{\boldsymbol{L}})=(\boldsymbol{P}^{*},\boldsymbol{L}^{(N)})$
is an optimal solution to $\mathscr{P}2$.

Next, we prove the cardinality of partly used slots in $\hat{\boldsymbol{L}}$
are less than $N$
\begin{align*}
\left|\bar{\mathcal{S}}\left(\hat{\boldsymbol{L}}\right)\right| & =\left|\bar{\mathcal{S}}\left(\boldsymbol{L}^{\left(N\right)}\right)\right|\stackrel{(a)}{\le}\text{\ensuremath{\sum_{n=1}^{N}}}\left|\bar{\mathcal{S}}_{n}\left(\boldsymbol{L}^{\left(N\right)}\right)\right|\\
 & \stackrel{(b)}{\le}\bar{\mathcal{S}}_{N}\left(\boldsymbol{L}^{\left(N\right)}\right)+\text{\ensuremath{\sum_{n=1}^{N-1}}}\left|\bar{\mathcal{S}}_{n}\left(\boldsymbol{L}^{\left(N-1\right)}\right)\right|\\
 & ...\\
 & \stackrel{(c)}{\le}\text{\ensuremath{\sum_{n=1}^{N}}}\left|\bar{\mathcal{S}}_{n}\left(\boldsymbol{L}^{\left(n\right)}\right)\right|\stackrel{(d)}{\le}N
\end{align*}
where (a) holds because 
\begin{align*}
\bar{\mathcal{S}}\left(\boldsymbol{L}^{\left(N\right)}\right) & =\bar{\mathcal{S}}\left(\boldsymbol{L}^{\left(N\right)}\right)\cap\bigcup_{n=1}^{N}\mathcal{T}_{n}\left(\boldsymbol{L}^{\left(N\right)}\right)\\
 & =\bigcup_{n=1}^{N}\left(\bar{\mathcal{S}}\left(\boldsymbol{L}^{\left(N\right)}\right)\cap\mathcal{T}_{n}\left(\boldsymbol{L}^{\left(N\right)}\right)\right)\\
 & =\bigcup_{n=1}^{N}\bar{\mathcal{S}}_{n}\left(\boldsymbol{L}^{\left(N\right)}\right)
\end{align*}
(d) holds because $\bar{\mathcal{S}}_{n}(\boldsymbol{L}^{(n)})\subseteq\{q_{k+1}\}$
according to (\ref{eq:allocation_lnt}), and (b) to (c) hold because
for any $n\in\{2,\cdots,N\}$ 
\begin{align*}
\text{\ensuremath{\sum_{i=1}^{n-1}}}\left|\bar{\mathcal{S}}_{i}\left(\boldsymbol{L}^{\left(n\right)}\right)\right| & =\text{\ensuremath{\sum_{i=1}^{n-1}}}\left|\bar{\mathcal{S}}\left(\boldsymbol{L}^{\left(n\right)}\right)\cap\mathcal{T}_{i}\left(\boldsymbol{L}^{\left(n\right)}\right)\right|\\
 & \stackrel{(e)}{=}\text{\ensuremath{\sum_{i=1}^{n-1}}}\left|\bar{\mathcal{S}}\left(\boldsymbol{L}^{\left(n\right)}\right)\cap\mathcal{T}_{i}\left(\boldsymbol{L}^{\left(n-1\right)}\right)\right|\\
 & \stackrel{(f)}{\le}\text{\ensuremath{\sum_{i=1}^{n-1}}}\left|\bar{\mathcal{S}}\left(\boldsymbol{L}^{\left(n-1\right)}\right)\cap\mathcal{T}_{i}\left(\boldsymbol{L}^{\left(n-1\right)}\right)\right|\\
 & =\text{\ensuremath{\sum_{i=1}^{n-1}}}\left|\bar{\mathcal{S}}_{i}\left(\boldsymbol{L}^{\left(n-1\right)}\right)\right|
\end{align*}
where (e) holds because $\mathcal{T}_{i}(\boldsymbol{L}^{(n)})\subseteq\mathcal{T}_{i}(\boldsymbol{L}^{(n-1)})$
due to $\boldsymbol{l}_{i}^{(n)}=\boldsymbol{l}_{i}^{(n-1)},\forall i\neq n$
in the $n$th round, and (f) holds because $\bar{\mathcal{S}}(\boldsymbol{L}^{(n)})\subseteq\bar{\mathcal{S}}(\boldsymbol{L}^{(n-1)})$
since the step 3 in Algorithm \ref{alg:Rounding-Algorithm} only modifies
$\boldsymbol{l}_{n}(t)$ for $t\in\bar{\mathcal{S}}_{n}(\boldsymbol{L}^{(n-1)})\subseteq\bar{\mathcal{S}}(\boldsymbol{L}^{(n-1)})$.

Finally, combining the results that $(\boldsymbol{P}^{*},\hat{\boldsymbol{L}})$
is an optimal solution to $\mathscr{P}2$ and $|\bar{\mathcal{S}}(\hat{\boldsymbol{L}})|\le N$,
we can deduce that $F(\bm{P}^{*},\hat{\boldsymbol{L}})-F^{*}\le N\delta\lambda$
according to Proposition \ref{prop:perf_cap_non_r}.

\bibliographystyle{IEEEtran}
\bibliography{IEEEabrv,StringDefinitions,JCgroup,ChenBibCV,JCgroup-bw}

\begin{thebibliography}{10}
\providecommand{\url}[1]{#1}
\csname url@samestyle\endcsname
\providecommand{\newblock}{\relax}
\providecommand{\bibinfo}[2]{#2}
\providecommand{\BIBentrySTDinterwordspacing}{\spaceskip=0pt\relax}
\providecommand{\BIBentryALTinterwordstretchfactor}{4}
\providecommand{\BIBentryALTinterwordspacing}{\spaceskip=\fontdimen2\font plus
\BIBentryALTinterwordstretchfactor\fontdimen3\font minus
  \fontdimen4\font\relax}
\providecommand{\BIBforeignlanguage}[2]{{%
\expandafter\ifx\csname l@#1\endcsname\relax
\typeout{** WARNING: IEEEtran.bst: No hyphenation pattern has been}%
\typeout{** loaded for the language `#1'. Using the pattern for}%
\typeout{** the default language instead.}%
\else
\language=\csname l@#1\endcsname
\fi
#2}}
\providecommand{\BIBdecl}{\relax}
\BIBdecl

\bibitem{WuXuZenNg:J21}
Q.~Wu, J.~Xu, Y.~Zeng, D.~W.~K. Ng, N.~Al-Dhahir, R.~Schober, and A.~L.
  Swindlehurst, ``A comprehensive overview on 5{G}-and-beyond networks with
  {UAV}s: From communications to sensing and intelligence,'' \emph{IEEE Journal
  on Selected Areas in Communications}, vol.~39, no.~10, pp. 2912--2945, 2021.

\bibitem{YaoWanXuXu:J20}
K.~Yao, J.~Wang, Y.~Xu, Y.~Xu, Y.~Yang, Y.~Zhang, H.~Jiang, and J.~Yao,
  ``Self-organizing slot access for neighboring cooperation in {UAV} swarms,''
  \emph{IEEE Trans. on Wireless Commun.}, vol.~19, no.~4, pp. 2800--2812, 2020.

\bibitem{ShuSaaDaoNa:J20}
D.~Shumeye~Lakew, U.~Saad, N.-N. Dao, W.~Na, and S.~Cho, ``Routing in flying ad
  hoc networks: A comprehensive survey,'' \emph{{IEEE} Commun. Surveys Tuts.},
  vol.~22, no.~2, pp. 1071--1120, 2020.

\bibitem{MouGaoLiuWu:J22}
Z.~Mou, F.~Gao, J.~Liu, and Q.~Wu, ``Resilient {UAV} swarm communications with
  graph convolutional neural network,'' \emph{{IEEE} J. Sel. Areas Commun.},
  vol.~40, no.~1, pp. 393--411, 2022.

\bibitem{ZhaZhaDiSon:J19}
S.~Zhang, H.~Zhang, B.~Di, and L.~Song, ``Cellular {UAV}-to-{X} communications:
  Design and optimization for multi-{UAV} networks,'' \emph{IEEE Trans. on
  Wireless Commun.}, vol.~18, no.~2, pp. 1346--1359, 2019.

\bibitem{ChaSadBet:J19}
U.~Challita, W.~Saad, and C.~Bettstetter, ``Interference management for
  cellular-connected {UAV}s: A deep reinforcement learning approach,''
  \emph{IEEE Trans. on Wireless Commun.}, vol.~18, no.~4, pp. 2125--2140, 2019.

\bibitem{ZhuGuoLiChe:J19}
X.~Zhong, Y.~Guo, N.~Li, Y.~Chen, and S.~Li, ``Deployment optimization of {UAV}
  relay for malfunctioning base station: Model-free approaches,'' \emph{{IEEE}
  Trans. Veh. Technol.}, vol.~68, no.~12, pp. 11\,971--11\,984, 2019.

\bibitem{AzaGerGarPol:J20}
M.~M. Azari, G.~Geraci, A.~Garcia-Rodriguez, and S.~Pollin, ``{UAV}-to-{UAV}
  communications in cellular networks,'' \emph{IEEE Trans. on Wireless
  Commun.}, vol.~19, no.~9, pp. 6130--6144, 2020.

\bibitem{HasKadAkh:J22}
M.~Z. Hassan, G.~Kaddoum, and O.~Akhrif, ``Interference management in
  cellular-connected internet of drones networks with drone-pairing and uplink
  rate-splitting multiple access,'' \emph{{IEEE} Internet Things J.}, vol.~9,
  no.~17, pp. 16\,060--16\,079, 2022.

\bibitem{RahHosHeDai:J22}
A.~Rahmati, S.~Hosseinalipour, Y.~Yapici, X.~He, I.~Guvenc, H.~Dai, and
  A.~Bhuyan, ``Dynamic interference management for {UAV}-assisted wireless
  networks,'' \emph{IEEE Trans. on Wireless Commun.}, vol.~21, no.~4, pp.
  2637--2653, 2022.

\bibitem{TanZhaHeZhu:J22}
W.~Tang, H.~Zhang, Y.~He, and M.~Zhou, ``Performance analysis of multi-antenna
  {UAV} networks with {3D} interference coordination,'' \emph{IEEE Trans. on
  Wireless Commun.}, vol.~21, no.~7, pp. 5145--5161, 2022.

\bibitem{ZenLyuZha:J19}
Y.~Zeng, J.~Lyu, and R.~Zhang, ``Cellular-connected {UAV}: Potential,
  challenges, and promising technologies,'' \emph{{IEEE} Wireless Commun.},
  vol.~26, no.~1, pp. 120--127, 2019.

\bibitem{ZhaZenZha:J18}
J.~Zhang, Y.~Zeng, and R.~Zhang, ``{UAV}-enabled radio access network:
  Multi-mode communication and trajectory design,'' \emph{{IEEE} Trans. Signal
  Process.}, vol.~66, no.~20, pp. 5269--5284, 2018.

\bibitem{JiYanSheXu:J20}
Y.~Ji, Z.~Yang, H.~Shen, W.~Xu, K.~Wang, and X.~Dong, ``Multicell edge coverage
  enhancement using mobile {UAV}-relay,'' \emph{{IEEE} Internet Things J.},
  vol.~7, no.~8, pp. 7482--7494, 2020.

\bibitem{LiChaCai:J20}
L.~Li, T.-H. Chang, and S.~Cai, ``{UAV} positioning and power control for
  two-way wireless relaying,'' \emph{IEEE Trans. on Wireless Commun.}, vol.~19,
  no.~2, pp. 1008--1024, 2020.

\bibitem{MaZhoQiaChe:J21}
T.~Ma, H.~Zhou, B.~Qian, N.~Cheng, X.~Shen, X.~Chen, and B.~Bai, ``{UAV}-{LEO}
  integrated backbone: A ubiquitous data collection approach for {B5G} internet
  of remote things networks,'' \emph{{IEEE} J. Sel. Areas Commun.}, vol.~39,
  no.~11, pp. 3491--3505, 2021.

\bibitem{LiChe:C22}
B.~Li and J.~Chen, ``Handover game for data transportation over dynamic {UAV}
  networks with predictable channels,'' in \emph{Proc. IEEE Global Commun.
  Conf.}, 2022, pp. 3724--3729.

\bibitem{HuCaiYuQin:J19}
Q.~Hu, Y.~Cai, G.~Yu, Z.~Qin, M.~Zhao, and G.~Y. Li, ``Joint offloading and
  trajectory design for {UAV}-enabled mobile edge computing systems,''
  \emph{{IEEE} Trans. Inf. Theory}, vol.~6, no.~2, pp. 1879--1892, 2019.

\bibitem{HuCaiLiuYu:J20}
Q.~Hu, Y.~Cai, A.~Liu, G.~Yu, and G.~Y. Li, ``Low-complexity joint resource
  allocation and trajectory design for {UAV}-aided relay networks with the
  segmented ray-tracing channel model,'' \emph{IEEE Trans. on Wireless
  Commun.}, vol.~19, no.~9, pp. 6179--6195, 2020.

\bibitem{AlsYuk:J21}
A.~Alsharoa and M.~Yuksel, ``Energy efficient {D2D} communications using
  multiple {UAV} relays,'' \emph{IEEE Trans. on Commun.}, vol.~69, no.~8, pp.
  5337--5351, 2021.

\bibitem{YouZha:J19}
C.~You and R.~Zhang, ``{3D} trajectory optimization in {Rician} fading for
  {UAV}-enabled data harvesting,'' \emph{IEEE Trans. on Wireless Commun.},
  vol.~18, no.~6, pp. 3192--3207, 2019.

\bibitem{SamShaAssNgu:J20}
M.~Samir, S.~Sharafeddine, C.~M. Assi, T.~M. Nguyen, and A.~Ghrayeb, ``{UAV}
  trajectory planning for data collection from time-constrained {IoT}
  devices,'' \emph{IEEE Trans. on Wireless Commun.}, vol.~19, no.~1, pp.
  34--46, 2020.

\bibitem{LiZhaZhaYan:J21}
B.~Li, S.~Zhao, R.~Zhang, and L.~Yang, ``Full-duplex {UAV} relaying for
  multiple user pairs,'' \emph{{IEEE} Internet Things J.}, vol.~8, no.~6, pp.
  4657--4667, 2021.

\bibitem{LevYapKutCai:J21}
R.~Levie, {\c{C}}.~Yapar, G.~Kutyniok, and G.~Caire, ``{RadioUNet}: Fast radio
  map estimation with convolutional neural networks,'' \emph{IEEE Trans. on
  Wireless Commun.}, vol.~20, pp. 4001--4015, 2021.

\bibitem{liuche:J23}
W.~Liu and J.~Chen, ``{UAV}-aided radio map construction exploiting environment
  semantics,'' \emph{IEEE Trans. on Wireless Commun.}, 2023,
  doi:{10.1109/TWC.2023.3241845}.

\bibitem{Fis:J81}
M.~L. Fisher, ``The lagrangian relaxation method for solving integer
  programming problems,'' \emph{Management science}, vol.~27, no.~1, pp. 1--18,
  1981.

\bibitem{Raz:T14}
M.~Razaviyayn, ``Successive convex approximation: Analysis and applications,''
  Ph.D. dissertation, University of Minnesota, 2014.

\bibitem{boyd2004convex}
S.~Boyd, S.~P. Boyd, and L.~Vandenberghe, \emph{Convex optimization}.\hskip 1em
  plus 0.5em minus 0.4em\relax Cambridge university press, 2004.

\bibitem{Ber:b15}
D.~Bertsekas, \emph{Convex optimization algorithms}.\hskip 1em plus 0.5em minus
  0.4em\relax Athena Scientific, 2015.

\bibitem{TR36814}
``Evolved universal terrestrial radio access ({E-UTRA}); further advancements
  for {E-UTRA} physical layer aspects,'' 3GPP, Tech. Rep. TR 36.814 (Release
  9), Mar. 2017.

\bibitem{AlhKanLar:J14}
A.~Al-Hourani, S.~Kandeepan, and S.~Lardner, ``Optimal {LAP} altitude for
  maximum coverage,'' \emph{IEEE Wireless Commun. Lett.}, vol.~3, no.~6, pp.
  569--572, 2014.

\bibitem{MozSadBen:J17}
M.~Mozaffari, W.~Saad, M.~Bennis, and M.~Debbah, ``Mobile unmanned aerial
  vehicles (uavs) for energy-efficient internet of things communications,''
  \emph{IEEE Trans. on Wireless Commun.}, vol.~16, no.~11, pp. 7574--7589,
  2017.

\bibitem{MicSte:14}
M.~Grant and S.~Boyd, ``{CVX}: Matlab software for disciplined convex
  programming, version 2.1,'' \url{http://cvxr.com/cvx}, Mar. 2014.

\bibitem{SebKarDig:J18}
J.~Sebastian, C.~Karakus, and S.~Diggavi, ``Approximate capacity of fast fading
  interference channels with no instantaneous {CSIT},'' \emph{IEEE Trans. on
  Commun.}, vol.~66, no.~12, pp. 6015--6027, 2018.

\end{thebibliography}

\end{document}